 \documentclass[11pt]{article}
 
\usepackage{amssymb,amsmath,amsfonts}
\usepackage{graphicx,color,enumitem}
\usepackage{mathrsfs}
\usepackage{tikz}
\usepackage{subcaption}
\usepackage{amsthm} 
\usepackage{bm}
\usepackage{comment}
\usepackage{float}
\usepackage{shuffle}
\usepackage[round]{natbib}
\usepackage[]{appendix}
\usepackage{xcolor}
\usepackage[a4paper, total={6in, 9in}]{geometry}
\usepackage{subcaption}

\RequirePackage[colorlinks,citecolor=blue!70!green,urlcolor=blue, linkcolor=blue!70!green]{hyperref}

\usepackage{tikz}
\tikzstyle{vertex}=[circle, draw, inner sep=2pt, fill=white]

\renewcommand{\d}{\mathrm{d}}
\newcommand{\e}{{\varepsilon}}

\newcommand{\E}{{\mathbb E}}
\newcommand{\Var}{{\mathbb V}\textup{ar}}

\renewcommand{\P}{{\mathbb P}}
\newcommand{\Q}{{\mathbb Q}}

\newcommand{\R}{{\mathbb R}}
\renewcommand{\S}{{\mathbb S}}
\newcommand{\N}{{\mathbb N}}
\newcommand{\W}{{\mathbb W}}
\newcommand{\X}{{\mathbb X}}
\newcommand{\Y}{{\mathbb Y}}

\newcommand{\Acal}{{\mathcal A}}

\newcommand{\Dcal}{{\mathcal D}}

\newcommand{\Fcal}{{\mathcal F}}

\newcommand{\Ical}{{\mathcal I}}

\newcommand{\Scal}{{\mathcal S}}

\newcommand{\Xx}{{X}}
\newcommand{\Ss}{{S}}

\newcommand{\halfshuffle}{{\widetilde \shuffle}}

\newcommand{\Span}{\textup{span}}

\newcommand{\fdot}{{\,\cdot\,}}

\DeclareMathOperator{\argmin}{argmin}

\DeclareMathOperator{\supp}{supp}

\newtheorem{theorem}{Theorem}

\newtheorem{assumption}[theorem]{Assumption}

\newtheorem{corollary}[theorem]{Corollary}

\theoremstyle{definition}
\newtheorem{definition}[theorem]{Definition}
\newtheorem{remark}[theorem]{Remark}
\newtheorem{example}[theorem]{Example}

\newtheorem{lemma}[theorem]{Lemma}
\newtheorem{notation}[theorem]{Notation}

\newtheorem{proposition}[theorem]{Proposition}

\numberwithin{equation}{section}
\numberwithin{theorem}{section}

\begin{document}

\title{Signature-based models: theory and calibration}
\author{Christa Cuchiero\thanks{Vienna University, Department of Statistics and Operations Research, Data Science Uni Vienna,
		Kolingasse 14-16 1, A-1090 Wien, Austria, christa.cuchiero@univie.ac.at}
	\and Guido Gazzani\thanks{Vienna University, Department of Statistics and Operations Research,
		Kolingasse 14-16 1, A-1090 Wien, Austria, guido.gazzani@univie.ac.at} 
	\and Sara Svaluto-Ferro\thanks{University of Verona, Department of Economics,
		Via Cantarane 24, 37129 Verona, Italy, sara.svalutoferro@univr.it\newline
		The authors gratefully acknowledge financial support by the FWF project I 3852
and through grant Y 1235 of the FWF START-program.
}}

\maketitle
\begin{abstract}

    We consider asset price models whose  dynamics are described by linear functions of the (time extended) signature of a primary underlying process, which can range from a (market-inferred) Brownian motion to a general multidimensional continuous semimartingale. The framework is universal in the sense that classical models can be approximated arbitrarily well and that the model's parameters can be learned from all sources of available data by simple methods. We provide conditions guaranteeing absence of arbitrage as well as tractable option pricing formulas for so-called sig-payoffs, exploiting the polynomial nature of generic primary processes. One of our main focus lies on  calibration, where we consider both  time-series and implied volatility surface data, generated from classical stochastic volatility models and also from S\&P~500 index market data. For both tasks the linearity of the model turns out to be the crucial tractability feature  which allows to get fast and accurate calibrations results.

\end{abstract}


\tableofcontents
\section{Introduction}

In the past few years data-driven models have successfully entered the area of stochastic modeling and mathematical finance. The paradigm of calibrating a few well interpretable parameters has changed to learning the model's characteristics as a whole, thereby exploiting all available sources of data. Thus highly parametric and overparametrized models  have gained more and more importance. 
On the one hand  this has opened the door to robust and more data-driven model
selection mechanisms, while on the other hand  model classes  still have to be chosen in a way to guarantee first principles from finance like ``no arbitrage''.
Relying on different universal approximation theorems  leads to different well-suited classes of dynamic processes that can serve both purposes.	

One class of such financial models are so-called \emph{neural stochastic differential equations} (SDEs) which are defined as It\^o-diffusions where the drift and the volatility function are parameterized via neural networks (see e.g.~\cite{GSSSZ:20, CKT:20, CRW:21}).
Another class of models, considered in \cite{PSS:20} and inspiring the current work,  are so-called \emph{Sig-SDEs}.  These are again It\^{o}-diffusions, however in this case the characteristics are  linear functions of the \emph{signature} (more precisely introduced below) of some (properly extended) driving Brownian motion.

We consider here a related approach, where the asset price model itself is parameterized as a linear function of the signature of a primary underlying process. This primary process can either be a classical driving signal, e.g.~a Brownian motion, but also a more general tractable stochastic model describing well observable quantities. 

Before going into the details of the current model framework, let us first explain  the mathematical significance of  \emph{signature}, a notion which 
goes back to \cite{C:77, C:57} and plays a particular important role in the context of rough path theory initiated by \cite{L:98}. Indeed, we consider here the \emph{time extended signature} of an $\mathbb{R}^d$-valued path which serves as linear regression basis for continuous path functionals, since
\begin{itemize}
\item it is \emph{point-separating}, as its final value uniquely determines the underlying path;
\item linear functions on the signature form an algebra of continuous functions (with respect to a certain variation distance) that contains 1. More precisely,  
every polynomial on the signature may be realized as a linear function via
the so-called shuffle product $\shuffle$.
\end{itemize}

The Stone-Weierstrass theorem therefore yields a Universal Approximation Theorem (UAT), telling that continuous  path functionals on compact sets 
can be uniformly approximated by a linear function
of the time extended signature.

Therefore signature-based methods provide
a non-parametric way to extract characteristic features (linearly) from time series data, which is essential in machine learning tasks in finance. This explains why these techniques become more and more popular in econometrics and mathematical finance, see e.g., \cite{BHLPW:20, KLP:20, PSS:20, LNP:20, NSSBWL:21, BHRS:21,  MH:21, AGTZ:22} and the references therein.

We consider here signature-based models with the goal to provide a data-driven, universal, tractable and easy to calibrate model
for a set of traded assets $S$. For the sake of exposition we  shall assume throughout that $S$ is one-dimensional. 
As already mentioned above the main ingredient of our modeling framework is  a $d$-dimensional \emph{primary  process} $X=(X_{t}^{1},\dots,X_{t}^{d})_{t \geq 0}$, in most case augmented with time $t$, where $X$ is supposed to be a continuous semimartingale.
We shall denote the time augmented process via
$(\widehat{X}_{t})_{t \geq 0}=(t,X_{t}^{1},\dots,X_{t}^{d})_{t \geq 0}$  and assume that its signature denoted by $\widehat{\mathbb{X}}$ serves a linear regression basis for $S$.
Indeed, $S$ is modeled\slash approximated via a process $S_{n}(\ell)$ defined as
\begin{align}\label{eq:Sn}
S_n(\ell)_t:=\ell(\widehat{\mathbb{X}}_t),
	\end{align}
	where $\ell$ is  a linear map of the signature of  $\widehat{X}$ up to degree 
	$n\in \mathbb{N}$ that has to be inferred from data (see Definition \ref{def:model} for further details). The attractiveness of this model class arises from several important features that we summarize in the sequel.

\begin{description}

\item[No arbitrage:] As shown in Section \ref{sec1}, the model of form \eqref{eq:Sn} can also be expressed in terms of stochastic integrals. From this representation 
conditions guaranteeing no-arbitrage can in turn be easily deduced.

\item[Universality:]  
As we argue in Example \ref{example:tilde:sv:1d}
classical stochastic volatility models driven by Brownian motions with sufficiently regular  coefficients
can be arbitrarily  well approximated
by models of form \eqref{eq:Sn}, when choosing  the driving Brownian motions (possibly modified with some drift) as primary process $X$. Locally in time this can  be inferred from the stochastic Taylor expansion  that  we state in Proposition \ref{prop214}, which gives quantitative estimates. It can also seen as a consequence of the Stone-Weierstrass theorem in the way formulated e.g.~in Lemma 5.2 in~\cite{BHRS:21} (see also \cite{CM:22}), however without a convergence rate.

\item[Tractable option pricing formulas for sig-payoffs:] 
By relying on the above UAT and in turn on approximations via so-called sig-payoffs of the form $L(\widehat{\mathbb{S}}_T ({\ell}) )$ (see also \cite{LNP:20}), where  $\widehat{\mathbb{S}}$ denotes the signature of $(\widehat{S}_t)_t:=(t, S_t)_t$ and $L$ a linear map,
this kind of approximate options pricing  reduces to the computation of the expected signature of $\widehat{X}$. Indeed, in Section \ref{sec32} we first show how to express  $\widehat{\mathbb{S}}_T ({\ell}) $ in terms of the signature of $\widehat{X}$ and then translate this to the computation of $\mathbb{E}_{\mathbb{Q}}[L(\widehat{\mathbb{S}}_T ({\ell}) \rangle]$ under some pricing measure $\mathbb{Q}$. These formulas can then be applied whenever $\mathbb{E}_{\mathbb{Q}}[\widehat{\mathbb{X}}_T]$ can be easily computed. This is the case for
highly generic primary processes  of the form
\begin{align*}
\d\widehat{X}_{t}=\mathbf{b}(\widehat{\mathbb{X}}_t) \d t + \sqrt{\mathbf{a}(\widehat{\mathbb{X}}_t)} \d B_{t},  
\end{align*}
where $\mathbf{b}$ and $\mathbf{a}$ are linear maps.
Indeed, as shown in \cite{CST:22} these processes can be seen 
as projections of  extended tensor algebra valued polynomial processes (introduced in \cite{CKT:12, FL:16}), which implies that 
the expected signature  can  be computed by solving  a  linear ODE. This ODE is usually infinite dimensional, but if $\widehat{X}$ is itself a polynomial process it becomes finite dimensional. We exemplify this polynomial process point of view by means of time extended correlated Brownian motions, a particular simple polynomial process, in Section \ref{secBM}.

Note that similarly to polynomial approximations the approximation of vanilla call and put options via sig-payoffs is not straightforward and can involve several problems  as we point out in Section \ref{sec:calib-sigpayoffs}. Nevertheless sig-payoffs can  be used for variance reduction techniques (see Section \ref{sec:MCvar}) and are interesting in their own right as  certain path dependent options like Asian forwards fall into this class.
 
\item[Calibration to time series data:] 
The tractability of the model class given by \eqref{eq:Sn} becomes particularly clear in view of calibration tasks. Indeed, when the goal is to calibrate to times series data of prices $(S_{t_i})_{i=1}^N$ or spot volatility\slash spot variance $(V_{t_i})_{i=1}^N$,
this task reduces to a simple linear regression. We illustrate this well working procedure on out-of-sample data generated from Heston, SABR and multivariate Black-Scholes models  in Section \ref{sec:calibration_ts}.

\item[Calibration to options:] 
When calibrating to the market's volatility surface (see Section~\ref{sec:calibration_options}), we exploit again the linearity of the model by precomputing Monte Carlo samples of $\widehat{\mathbb{X}}$ and then performing a standard optimization to find the parameters of the linear map $\ell$. By initializing the parameters of $\ell$ appropriately
this optimization task actually becomes a convex problem, which makes it particularly tractable.
On simulated and real market data (S$\&$P 500 index) we show that
 a full calibration to the  volatility surface, in particular when using time dependent parameters, is  highly accurate  and very fast.
 Note that \emph{sampling from the calibrated model} is also particularly easy since it just means computing a scalar product between the parameters and the trajectory of the signature.\footnote{The codes which we used to generate the results in Section \ref{sec:calibration}  are available at  \href{https://github.com/GuidoGazzani-ai/sigsde_calibration}{https://github.com/GuidoGazzani-ai/sigsde$\_$calibration}.}

\end{description}

Let us remark that in \eqref{eq:Sn} the signature process of the time extended primary process $\widehat{X}$ could be replaced by other  processes representing the features of $S$. Indeed, instead of $\widehat{\mathbb{X}}$ one could for instance use randomized signature as introduced in \cite{CGGOT:21a, CGGOT:21b}  and work with a linear model of this randomized signature. Not all of the above properties, in particular the pricing formulas of the sig-payoffs, could then be preserved, but the tractability properties for calibration would remain valid and could potentially be speeded up even further, if the analog of $\widehat{\mathbb{X}}$ can be precomputed faster. We shall exploit this in a follow-up paper on VIX-calibration.

\subsection{Relation to the literature}

Let us here briefly summarize some relevant and recent literature on calibration of (data-driven) asset price models.
Closely related to the current work is the inspiring paper by \cite{PSS:20} (see also \cite{P:20}). Indeed, the class of Sig-SDEs
considered there can be embedded
in our framework by choosing a properly extended one-dimensional Brownian
motion (namely a lead-lag transformation) as primary  process. This is due to the fact that \eqref{eq:Sn} can also be expressed 
in terms of stochastic integrals as stated in Proposition \ref{prop1}.
While in \cite{PSS:20} calibration was considered for certain options within the Black-Scholes model, we here provide calibration results to both, time-series data and volatility surfaces generated from classical stochastic volatility models and also from real market data. A crucial difference is that the model in \cite{PSS:20} is calibrated to sig-payoffs that approximate vanilla call and puts. As we encountered some problems with this procedure (as outlined in Appendix \ref{sec:calib-sigpayoffs}), we rely on Monte Carlo pricing which is fast due to the possibility to precompute all samples of $\widehat{\mathbb{X}}$ in advance.

Other promising calibration results, also in view of joint calibration of SPX and VIX options, have for instance been achieved in the following recent papers \cite{GSSSZ:20, GLOW:22, GJ:21, GJ:22, JMP:21, RZ:21, R:22, V:22}, see also the references therein.

The remainder of the article is structured as follows: in Section \ref{sec:sig} we recall important concepts concerning the signature of continuous semimartingales, while Section \ref{sec:model} builds on these results to introduce our model framework. Section \ref{sec:calibration} is then dedicated to present our calibration results.

\section{Signature: definition and properties}\label{sec:sig}

We start by introducing basic notions related to the definition of the  signature of an $\mathbb{R}^{d}$-valued continuous semimartingale. For similar introductions to the concept of signature we refer e.g.~to
Section~2.2. in \cite{BHRS:21}.

\subsection{Basic notions}
For each $n \in \mathbb{N}_0$ consider the $n$-fold tensor product of $\mathbb{R}^{d}$ given by
\begin{equation*}
(\mathbb{R}^{d})^{\otimes 0}:=\mathbb{R}, \qquad (\mathbb{R}^{d})^{\otimes n}:=\underbrace{\mathbb{R}^{d}\otimes\cdots\otimes\mathbb{R}^{d}}_{n}.
\end{equation*}
	For $d\in \N$, we define the extended tensor algebra on $\mathbb{R}^{d}$ as 
	\begin{equation*}
		T((\mathbb{R}^{d})):=\{\textbf{a}:=(a_{0},\dots,a_{n},\dots) : a_{n}\in(\mathbb{R}^{d})^{\otimes n}\}.
	\end{equation*}
	Similarly we introduce the truncated tensor algebra of order $N \in \mathbb{N}$ 
		\begin{equation*}
		T^{(N)}(\mathbb{R}^{d}):=\{\textbf{a}\in T((\mathbb{R}^{d})) : a_{n}=0, \forall n>N\},
	\end{equation*}
	and the tensor algebra
$
		T(\mathbb{R}^{d}):=\bigcup_{N\in \N}T^{(N)}(\mathbb{R}^{d}).
$
Note that $T^{(N)}(\R^{d})$ has dimension $\sum_{i=0}^{N}d^{i}=(d^{N+1}-1)/(d-1)$.

	For each $\textbf{a},\textbf{b}\in T((\mathbb{R}^{d}))$ and  $\lambda\in\R$ we set 
	\begin{align*}
		\textbf{a}+\textbf{b}&:=(a_{0}+b_{0},\dots,a_{n}+b_{n},\dots),\\
		\lambda \textbf{a}&:= (\lambda a_{0},\dots, \lambda a_{n},\dots),\\
		\textbf{a}\otimes \textbf{b}&:=(c_{0},\dots, c_{n},\dots),
	\end{align*}
where $c_{n}:=\sum_{k=0}^{n}a_{k}\otimes b_{n-k}$. Observe that $(T((\mathbb{R}^{d})),+,\cdot,\otimes)$ is a real non-commutative algebra with neutral element \textbf{\textup{1}}$=(1,0,\dots,0,\dots)$.

	 For a multi-index $I:=(i_1,\ldots,i_n)$  we set $|I|:=n$. We also consider the empty index $I:=\emptyset$ and set $|I|:=0$. If $n\geq 1$ or $n\geq 2$ we set $I':=(i_1,\ldots,i_{n-1})$, and $I'':=(i_1,\ldots,i_{n-2})$, respectively. We also use the notation
	$$\{I\colon|I|=n\}:=\{1,\ldots,d\}^n,$$
	omitting the parameter $d$ whenever this does not introduce ambiguity. Observe that multi-indices can be identified with words, as it is done for instance in \cite{LNP:20}.

Next, for each $|I|\geq 1$ we set
	$$e_I:=e_{i_1}\otimes\cdots\otimes e_{i_n},$$
	where $e_1, \ldots, e_d$ denote the canonical basis vectors of $\mathbb{R}^d$.
	Observe that the set $\{e_I\colon |I|=N\}$ is an orthonormal basis of $(\mathbb{R}^{d})^{\otimes N}$.
	Denoting by $e_\emptyset$ the basis element corresponding to $(\R^d)^{\otimes 0}$, each element of $\textbf{a}\in T((\R^d))$ can thus be written as
	$$\textbf{a}=\sum_{|I|\geq 0}a_Ie_I,$$
	for some $a_I\in \R$.
	Finally, for each $\textbf{a}\in T(\R^d)$ and each $\textbf{b}\in T((\R^d))$ we set
		$$\langle \textbf{a},\textbf{b}\rangle:=\sum_{|I|\geq 0}\langle a_{I},b_{I}\rangle.$$
Observe in particular that $b_I=\langle e_I,\textbf{b}\rangle$.	

Throughout the paper we consider a filtered probability space $(\Omega, \Fcal, (\Fcal_t)_{t\geq0},\P)$. Whenever not specified, stochastic processes are always supposed to be defined there.
We are now ready to define the signature of an $\mathbb{R}^d$-valued continuous semimartingale.

\begin{definition}
	Let $(X_t)_{t\in[0,T]}$ be a continuous $\mathbb{R}^{d}$-valued semimartingale. The \emph{signature of $X$} is the $T((\R^d))$-valued process  $(s,t)\mapsto \X_{s,t}$ whose components are recursively defined as

\begin{equation*}
	\langle e_{\emptyset},\mathbb{X}_{s,t}\rangle:=\textup{1}, \qquad \langle e_{I}, \mathbb{X}_{s,t}\rangle:=\int_{s}^{t}\langle e_{I'},\mathbb{X}_{s,r}\rangle\circ \mathrm{d} X_{r}^{i_{n}},
\end{equation*}
for each $I=(i_1,\ldots, i_n)$ , $I'=(i_1,\ldots, i_{n-1})$ and $0\leq s\leq t\leq T$, where $\circ$ denotes the Stratonovich integral. 
Its projection $\X^N$ on $T^{(N)}(\mathbb{R}^{d})$ is given by
\begin{equation*}
	\mathbb{X}_{s,t}^{N}=\sum_{|I|\leq N} \langle e_{I}, \mathbb{X}_{s,t}\rangle e_{I}
\end{equation*}
and is called \emph{signature of $X$ truncated at level $N$}. If $s=0$, we use the notation $\X_t$ and $\X_t^N$, respectively.
\end{definition}
Observe that the signature of $X$ and the signature of $X-c$ coincides for each $c\in \R$. Moreover, with an equivalent notation we can write
\begin{align*}
    \X_t&=\bigg(1,\int_0^t1\circ \d X_s^1,\ldots, \int_0^t1\circ \d X_s^d,\int_0^t\bigg(\int_0^s 1 \circ \d X_r^1\bigg)\circ \d X_s^1,\\
&\qquad\int_0^t\bigg(\int_0^s 1 \circ \d X_r^1\bigg)\circ \d X_s^2,
\ldots,
\int_0^t\bigg(\int_0^s 1 \circ \d X_r^d\bigg)\circ \d X_s^d,\ldots\bigg).
\end{align*}
Using It\^o integrals this can thus be rewritten as 
\begin{align*}
    \X_t&=\bigg(1,X_t^1-X^1_0,\ldots, X_t^d-X_0^d,\int_0^t (X_s^1-X^1_0)  
\d X_s^1+\frac 1 2 [X^1]_t,\\
&\qquad\int_{0}^{t} (X_s^1-X^1_0)  
\d X_s^2+\frac 1 2 [X^1,X^2]_t,
\ldots,
\int_{0}^{t} (X_s^d-X^d_0)  
\d X_s^d+\frac 1 2 [X^d]_t,\ldots\bigg),
\end{align*}
where the square brackets denote the quadratic variation and covariation processes. In particular, by  the definition of the Stratonovich integral and It\^o's formula we can write
$$\langle e_{(i,i)},\X_t\rangle
=\int_0^t\bigg(\int_0^s 1\circ \d X^i_r\bigg)\circ \d X^i_s
=\frac 1 2 (X_t^i-X_0^i)^2=\langle e_i,\X_t\rangle^2,$$
showing that the quadratic expression on the right hand side has a linear representation. This property generalizes to every polynomial function. For the precise statement we first need to introduce the shuffle product. 

\begin{definition}\label{shuffle-product}
	For every two multi-indices $I:=(i_1,\ldots,i_n)$ and $J:=(j_1,\ldots,j_m)$ the \emph{shuffle product} is defined recursively as
	\begin{align*}
		e_{I}\shuffle e_{J}:= (e_{I'}\shuffle e_{J})\otimes e_{i_n}+(e_{I}\shuffle e_{J'})\otimes e_{j_m},
	\end{align*}
	with $e_{I}\shuffle e_{\emptyset}:= e_{\emptyset}\shuffle e_{I}= e_{I}$. It extends to $\textbf{a},\textbf{b}\in T(\R^d)$ as
	$$\textbf{a}\shuffle\textbf{b}=\sum_{|I|,|J|\geq0}a_Ib_J (e_I\shuffle e_J).$$
\end{definition}
Observe that $(T(\R^d),+,\shuffle)$ is a commutative algebra, which in particular means that the shuffle product is associative and commutative.

The proof of the rough paths version of the next result for  can be found for instance in \cite{R:58} or \cite{LCL:07}. 
\begin{proposition}[Shuffle property]\label{shuffle-property}
		Let $(X_t)_{t\in[0,T]}$ be a continuous $\mathbb{R}^{d}$-valued semimartingale and $I,J$ be two multi-indices. Then
	\begin{equation}\label{eqn5}
		\langle e_{I},\mathbb{X}\rangle \langle e_{J}, \mathbb{X}\rangle =\langle e_{I}\shuffle e_{J}, \mathbb{X}\rangle.
	\end{equation}
\end{proposition}
\begin{proof}
	 The result follows by induction using the chain rule for Stratonovich integrals.
\end{proof}

\begin{example}
	Let $(X_t)_{t\in[0,T]}$ be a continuous $\mathbb{R}^{d}$-valued semimartingale with 
	$X_{0}=0$. Then the (Stratonovich) integration by parts formula yields, for any $i,j\in\{1,\dots,d\}$
	\begin{align*}
	\langle e_{i},\mathbb{X}_{T}\rangle \langle e_{j},\mathbb{X}_{T}\rangle
	=X_{T}^{i}X_{T}^{j}&=\int_{0}^{T}X_{t}^{i}\circ{\d}X_{t}^{j}+\int_{0}^{T}X_{t}^{j}\circ{\d}X_{t}^{i},\\
	&= \langle e_{i}\otimes e_{j},\mathbb{X}_{T}\rangle +  \langle e_{j}\otimes e_{i},\mathbb{X}_{T}\rangle\\
	&= \langle e_{i}\shuffle e_{j}, \mathbb{X}_{T}\rangle.
	\end{align*}
\end{example}

\begin{example}\label{ex2}
	Let $(X_t)_{t\in[0,T]}$ be a continuous $\mathbb{R}$-valued semimartingale with 
	$X_{0}=0$. Then $\langle e_1,\X_T\rangle=X_T-X_0$. Since
	$$\underbrace{e_1\shuffle\ldots\shuffle e_1}_{k\text{-times}}=k! \underbrace{e_1\otimes\ldots\otimes e_1}_{k\text{-times}},$$
	  Proposition~\ref{shuffle-property} yields
	\begin{align*}
	\X_T=(1,X_T-X_0,\frac 1 2(X_T-X_0),\frac 1 {3!}(X_T-X_0)^3,\ldots,\frac 1 {k!}(X_T-X_0)^k,\ldots).
	\end{align*}
\end{example}
We recall now an important property of the signature. The result is known in the rough paths literature (see for instance \cite{BL:2016}), but can also be proved directly in the simpler situation of a continuous semimartingale  that contains time as strictly monotone component.

\begin{lemma}[Uniqueness of the signature]\label{uniqueness_sig}
	Let $(X_t)_{t\in[0,T]}$ and $(Y_t)_{t\in[0,T]}$ be two continuous $\R^{d}$-valued semimartingales with $X_0=Y_0=0$. Set $\widehat X_t:=(t,X_t)$, $\widehat Y_t:=(t,Y_t)$ and let $\widehat \X$ and $\widehat \Y$ be the corresponding signature processes. Then $\widehat\X_T=\widehat\Y_T$  if and only if $X_t=Y_t$ for each $t\in[0,T]$. 
\end{lemma}
\begin{proof}
Denote by $0,\ldots, d$ the indices of $\widehat X$ and $\widehat Y$, where 0 corresponds  to the time component.
	The claim follows by noticing that $\langle (e_i\shuffle (e_0^{\otimes k}))\otimes e_0,\widehat \X_T\rangle=\int_0^TX_t^i \frac {t^k}{k!} \d t$, which implies that $\widehat \X_T$ uniquely determines $\widehat{X}_t$ for every $t \in [0,T]$.
\end{proof}

\begin{remark}
  Observe that $(t^k/k!)_{k\in \N_0}$ is a basis of $L^2([0,T],\d t)$ and 
$$\langle (e_i\shuffle (e_0^{\otimes k}))\otimes e_0,\widehat\X_T\rangle
=\int_0^TX_t^i \frac {t^k}{k!} \d t$$
is the corresponding projection of $X^i$ on $t^k/k!$. This property can be used to explicitly construct a sequence of polynomials with random coefficients on $[0,T]$ converging almost surely to  $X^i$ on $L^2([0,T],\d t)$. It also establishes a link between signature-based methods and quantization methods (see for instance \cite{LP:02,PP:05} and \cite{BCJ:21, T:21} for recent advances).
\end{remark}

Let us now introduce the definition of the half-shuffle product (see also \cite{EP:15}, \cite{DLP:20}, where it is usually denoted with $\prec$). Given an $\R^d$-valued continuous semimartingale, this operation permits to write the signature of $\X^N$  as a linear map of $\X$. This result will be useful later when we will need to compute the signature of a model given by a linear combination of terms of $\widehat \X^N$ (see Section~\ref{sig-payoffs}).

\begin{definition}\label{def:halfshuffle}
Set $I=(i_1,\ldots,i_n)$ and $J=(j_1,\ldots,j_m)$ for some $n\in \N$ and $m>0$. We define the half-shuffle $\halfshuffle$ as $e_I\halfshuffle e_\emptyset:=0$ and 
\begin{equation*}
	e_{I}\halfshuffle e_{J}:=(e_{I}\shuffle e_{J'})\otimes e_{j_m}.
\end{equation*}
\end{definition}

\begin{lemma}\label{lem1}
For each multi-indices $I=(i_1,\ldots,i_n)$ and $J=(j_1,\ldots,j_m)$ it holds
\begin{equation}\label{eqn6}
\int_0^t \langle e_{I}, \X_s\rangle \circ \d\langle e_{J},\X_s\rangle
=\langle e_I\halfshuffle e_{J},\X_t\rangle,
\end{equation}
for each $t\geq0$.
\end{lemma}
\begin{proof}
By definition of the Stratonovich integral, Proposition~\ref{shuffle-property}, and by the definition of the signature  we know that 
$$
	\int_0^t \langle e_{I}, \X_s\rangle \circ \d\langle e_{J},\X_s\rangle
	=\int_0^t \langle e_{I}, \X_s\rangle  \langle e_{J'},\X_s\rangle\circ \d\langle e_{j_m},\Xx_s\rangle
	=\langle (e_{I}\shuffle e_{J'})\otimes e_{j_m},\X_t\rangle.
$$
\end{proof}
In order to combine the value of the signature on different time intervals  Chen's identity going back to \cite{C:57, C:77} turns out to be fundamental.  For the reader convenience, we propose here a direct proof using the definition of Stratonovich integrals.
\begin{lemma}[Chen's identity]\label{lem5}
    Let $(X_t)_{t\in[0,T]}$ be an $\R^d$-valued semimartingale. Then 
    \begin{equation*}
        \X_{s,t}=\X_{s,u}\otimes\X_{u,t}
    \end{equation*}
    for each $s\leq u\leq t\leq T$. This can equivalently be written as
    $$\langle e_I,\X_{s,t}\rangle=\sum_{e_{I_1}\otimes e_{I_2}=e_I}\langle e_{I_1},\X_{s,u}\rangle\langle e_{I_2},\X_{u,t}\rangle,$$
    for each multi-index $I$.
\end{lemma}
\begin{proof}
We proceed by induction on the length of the multi-index $I$. For $|I|=0$ we have $I=\emptyset$ and the statement is clear. Suppose now that the claim holds for each $|J|<n$ and set $I=(i_1,\ldots,i_n)$.
Then, applying Chen's identity to $\langle e_{I'},\mathbb{X}_{s,r}\rangle $ yields
\begin{align*}
    \langle e_{I},\mathbb{X}_{s,t}\rangle&=\int_{s}^{t}\langle e_{I'},\mathbb{X}_{s,r}\rangle\circ \d {X}_{r}^{i_n}\\
    &=\int_{s}^{u}\langle e_{I'},\mathbb{X}_{s,r}\rangle\circ \d  {X}_{r}^{i_n}+\int_{u}^{t}\sum_{e_{I_{1}'}\otimes e_{I_{2}'}=e_{I'}}\langle e_{I_{1}'},\mathbb{X}_{s,u}\rangle \langle e_{I_{2}'},\mathbb{X}_{u,r}\rangle\circ  \d {X}_{r}^{i_n}\\
        &=\langle e_{I},\mathbb{X}_{s,u}\rangle
        +\sum_{e_{I_{1}'}\otimes e_{I_{2}'}=e_{I'}}\langle e_{I_{1}'},\mathbb{X}_{s,u}\rangle \langle e_{I_{2}'}\otimes e_{i_n},\mathbb{X}_{u,t}\rangle\\
&=\sum_{e_{I_{1}}\otimes e_{I_{2}}=e_{I}} \langle e_{I_1},\X_{s,u}\rangle\langle e_{I_2},\X_{u,t}\rangle,
\end{align*}
whence the claim follows.
\end{proof}

\subsection{Universal approximation theorem}

This section is devoted to the universal approximation result mentioned in the introduction.
Loosely speaking, it states that every quantity of the form
$$f\Big((\widehat\X_t^2)_{t\in[0,T]}\Big)$$
for some continuous map $f$ and some $T>0$
can be approximated arbitrarily well on compact sets by 
linear functions of the signature of the form $\langle \ell,\widehat \X_T\rangle$ where $\ell \in T(\mathbb{R}^d)$. 
Observe that the latter just involves the \emph{final value} $\widehat \X_T$ of $\widehat \X$, instead of its whole trajectory. Different versions of this result, e.g.~for finite variation paths or for continuous functions depending on the whole signature (instead of level 2), are available in the literature (see for instance Theorem~3.1 in \cite{LLN:13}, Theorem~1 in \cite{KO:19}, Proposition~4.5 in \cite{LNP:20} and Section~3 in \cite{CPS:22}). For completeness and to keep the paper self-contained we prove it here in the current context of continuous semimartingales and continuous functions of the respective paths lifted up to order 2.

Fix $T>0$, consider a continuous $\R^d$-valued semimartingale $(X_t)_{t\in[0,T]}$, and let $\widehat X_t:=(t, X_t)$. Denote by $0,\ldots, d$ the indices of $\widehat X$, where 0 corresponds  to the time component.
For each $N \in \mathbb{N}$ define the set 
$$\Scal^{(N)}:=\{(\widehat \X_t^N)_{t \in [0,T]}(\omega)\colon \omega\in \Omega\},$$
which, without loss of generality (passing to a subset of $\Omega$ of measure $1$), corresponds to a set of signature paths of $\widehat X$ up to time $T$.

The next lemma states that the signature of a continuous semimartingale coincides with the so-called Lyons lift (see for instance Theorem~2.2.1 in \cite{L:98}). This in particular implies that higher order terms of the signature can be defined pathwise starting from the trajectories of $\widehat\X^2$.
\begin{lemma}
For each $N\in \N$ there exists a map $S^{(N)}:\Scal^{(2)}\to \Scal^{(N)}$ 
 such that
\begin{equation}\label{eqn8}
\widehat\X^N=S^{(N)}(\widehat\X^2)
\end{equation}
almost surely.
\end{lemma}
\begin{proof}
The claim follows from Exercise~17.2 and Theorem~9.5  in \cite{FV:10}. 
\end{proof}
Without loss of generality (passing again to a subset of $\Omega$ of measure $1$) assume that  conditions \eqref{eqn5}, \eqref{eqn6},  \eqref{eqn8} and
$$\langle e_I\otimes e_0,\widehat \X_t(\omega)\rangle
=\int_0^t \langle e_I,\widehat \X_s(\omega)\rangle \d s $$
holds for each $\omega \in \Omega$ and each multi-index $I$.
Consider a generic distance $d_{\Scal^{(2)}}$ on the set of trajectories given by $\Scal^{(2)}$,
 with respect to which the map from $\Scal^{(2)}$ to $\R$ given by
$$\hat{\textbf x}^2\mapsto\langle e_I,S^{(|I|)}(\hat{\textbf x}^2)_t\rangle$$
is continuous for each multi-index $I$ and every $t \in [0,T]$.

\begin{theorem}[Universal approximation theorem]\label{universality}
	Let $K$ be a compact subset of $\Scal^{(2)}$ and  consider a continuous map $f: K\to \R$.\footnote{Compactness and continuity are defined with respect to $d_{\Scal^{(2)}}$.} Then for every $\varepsilon>0$ there exists some $\ell\in T(\R^{d})$ such that
	\begin{equation*}
	\sup_{(\widehat\X_t^2)_{t \in [0,T]}\in K}	\lvert f((\widehat \X^2_t)_{t \in [0,T]})-\langle \ell, \widehat \X_{T}\rangle\lvert <\varepsilon,
	\end{equation*}
	almost surely.
\end{theorem}
\begin{proof}
The result follows by an application of the Stone-Weierstrass theorem. To check the needed properties, observe that by Proposition~\ref{shuffle-property} the set 
$$A:=\Span\{\hat{\textbf{x}}^2\mapsto\langle e_I,S^{(|I|)}(\hat{\textbf{x}}^2)_T\rangle\colon I\in\{0,\ldots,d\}^{|I|}\}$$
is an algebra of continuous maps from $(K,d_{\Scal^{(2)}})$ to $\R$. Choosing $I=\emptyset$ we get that $A$ is vanishing nowhere and we just need to check that it is point separating in $\Scal^{(2)}$.    Lemma~\ref{uniqueness_sig} yields the result for $\langle e_1,\hat{\textbf{x}}^2\rangle,\ldots, \langle e_d,\hat{\textbf{x}}^2\rangle$. Due to Lemma~\ref{lem1}, this then applies also to the remaining components of $\hat{\textbf{x}}^2$.
\end{proof}

\begin{remark}\label{rem3}
\begin{enumerate}
\item 
Observe that the assumptions on $d_{\Scal^{(2)}}$ guarantee that every map of the form
$$f\Big((\widehat\X_t^2)_{t\in[0,T]}\Big):=\tilde f(\langle e_{I_1},\widehat \X_{t_1}\rangle,\ldots,\langle e_{I_n},\widehat \X_{t_n}\rangle),$$
for some $\tilde f \in C(\R^n)$, is continuous with respect to $d_{\Scal^{(2)}}$.
    \item \label{rem3ii}
The rough paths literature provides several metrics $d_{\Scal^{(N)}}$ with respect to which the maps $S^{(N)}:(\Scal^{(2)},d_{\Scal^{(2)}})\to(\Scal^{(N)},d_{\Scal^{(N)}})$ are almost surely continuous on bounded sets. Among them one can for instance consider
$$d_{\Scal^{(N)}}(\hat{\textbf{x}},\hat{\textbf{y}})
:=\max_{k\in \{1,\ldots,N\}}\sup_\Dcal\bigg(\sum_{t_i\in \Dcal}
|\pi_k(\hat{\textbf x}_{t_{i-1},t_i}-\hat{\textbf y}_{t_{i-1},t_i})|^{p/k}\bigg)^{k/p}.$$
where $p\in (2,3)$, $\pi_k(\hat{\textbf{x}}):=\sum_{|I|=k}\langle e_I,\hat{\textbf{x}}\rangle e_I$, and $\Dcal$ denotes the set of all partitions of $[0,T]$, 
see Corollary 9.11, Definition 8.6 and Theorem 8.10 in \cite{FV:10} for the result or  \cite{CPS:22} for a more detailed explanation. Note that
$$\pi_1(\hat{\textbf x}_{t_{i-1},t_i})
=\sum_{i=1}^d\langle e_i,\hat{\textbf{x}}_{t_i}\rangle-\langle e_i,\hat{\textbf{x}}_{t_{i-1}}\rangle$$ and $$\pi_2(\hat{\textbf x}_{t_{i-1},t_i})
=
\sum_{i, j=1}^d\langle e_{ij},\hat{\textbf{x}}_{t_{i}}\rangle-\langle e_{ij},\hat{\textbf{x}}_{t_{i-1}}\rangle
-\langle e_{i},\hat{\textbf{x}}_{t_{i-1}}\rangle(\langle e_{j},\hat{\textbf{x}}_{t_{i}}\rangle-\langle e_{j},\hat{\textbf{x}}_{t_{i-1}}\rangle).$$

\item The universal approximation theorem can also be stated without constructing the involved spaces based on a fixed semimartingale. Since this would require a more advanced knowledge of rough paths theory, we address the interested reader to  
 \cite{L:98} or \cite{CF:19} for uniqueness and continuity of $S^{(N)}:(\Scal^{(2)},d_{\Scal^{(2)}})\to(\Scal^{(N)},d_{\Scal^{(N)}})$  and to \cite{CPS:22} for a complete formulation and proof of the resulting statement in the case of càdlàg semimartingales.
 \item In other versions of the universal approximation theorem the map $f$ is defined on a subset of  $T((\R^d))$-valued paths, which includes the realisations of $\X$, see for instance Theorem~3.1 in \cite{LLN:13} and Proposition~4.5 in \cite{LNP:20}.
 Here, the critical issues are to define the involved metric explicitly and to identify compact subsets of $T((\R^d))$-valued paths.
 \end{enumerate}
\end{remark}

\subsection{Stochastic Stratonovich Taylor expansion}

In addition to the (global) universal approximation theorem formulated above, we now also state a well-known quantitative 
approximation result for solutions of SDEs based on the (stochastic) Stratonovich Taylor expansion, see \cite{KP:92}.
We also refer to \cite{LO:14}  for approximations of non-anticipative path functionals of SDEs via Chen-Fliess series of iterated stochastic integrals. In contrast to \cite{KP:92} and \cite{LO:14}, where the $L^2$-error is considered (see Remark~\ref{rem4} below), we here directly
provide a `convergence rate in probability' without requiring second moment estimates on the coefficients of the stochastic Taylor expansion. To this end we rely on the deterministic estimates provided for rough differential equation as in Section 10 in \cite{FV:10}.

\begin{proposition}\label{prop214}
Let $\widehat X$ be a time extended $d$-dimensional Brownian motion and  $(Y_t)_{t\in[0,T]}$ be an $D$-dimensional strong solution of the SDE 
$$\d Y_t=\sum_{j=0}^d\Psi_j(Y_t)^\top\circ\d \widehat X_t^j,$$
for some smooth $\Psi_j:\R^D\to\R^D$.
For each $I$ let $\Psi_I:\R^D\to\R^D$ be the map whose $k$-th component is given by 
$$\Psi_I^k(y)=\Psi_{i_1}^\top \nabla (\Psi_{i_2}^\top\nabla \cdots (\Psi_{i_{n}}^\top e_k))(y).$$
Then, for each $m\geq 2$ and  $\e>0$ there is a 
constant $C(m,\e)$ such that
$$\P\bigg(\bigg|Y_t-Y_0
-\sum_{0<|I|\leq m}
\Psi_I(Y_0)\langle e_I,\widehat  \X_t\rangle\bigg|
>  C(m,\e)t^{(m+1)/2}\bigg)<\e.$$
\end{proposition}

\begin{proof}For each $K>0$ consider the stopping times
$\tau_K:=\inf\{t\geq0\colon |Y_t|\geq K\}.$ For each $j$ and $K$ fix smooth $\Psi_j^K:\R^D\to\R^D$ such that $\Psi_j^K(y)=\Psi_j(y)$ for each $|y|\leq K$ and $\supp(\Psi_j^K)\subseteq\{y\in\R^D\colon |y|\leq K+1\}$. Choose  $\gamma:=m+1$ and $2<p\leq \gamma$.
Exercise 17.2, Theorem~17.3  and
Corollary~10.15 in \cite{FV:10}\footnote{Observe that in \cite{FV:10} a vector field $\Psi_j$ is identified with the first order operator $f\mapsto \Psi_j\cdot\nabla f$ (see for instance the discussion on page 125). The Euler scheme $\mathcal E_\Psi$ is defined in Definition~10.1.}   applied to
$$\d Y_{t}=\sum_{j=0}^d \Psi_j^K(Y_t)\circ\d \widehat X_{t\land {\tau_k}}^j$$
yields the existence of a constant $C$ such that
$$
\bigg|Y_t-Y_0
-\sum_{0<|I|\leq m}
\Psi_I(y_0)\langle e_I,\widehat  \X_t\rangle\bigg|1_{\{\tau_K\geq T\}}
\leq   |\Psi|_{\text{Lip}^{\gamma-1}}^\gamma C\|\widehat \X^2\|^\gamma_{1/p-\text{Höl};[0,T]}t^{\gamma/p},
$$
almost surely, where
$|\Psi|_{\text{Lip}^{\gamma-1}}^\gamma$ and $\|\cdot\|^\gamma_{1/p-\text{Höl};[0,T]}
$ are
 defined in Definition~10.2 and Equation (8.2) of \cite{FV:10}, respectively. Observe that $\gamma/p\leq (m+1)/2$ and
$$|\Psi|_{\text{Lip}^{\gamma-1}}^\gamma\leq 
\sup_{k\leq m+1, |y|\leq K+1}|\Psi^{(k)}(y)|^\gamma
\leq
\bigg(1+\sup_{k\leq m+1, |y|\leq K+1}|\Psi^{(k)}(y)|\bigg)^{m+1}.$$
By the discussion at the beginning of Section~3 in \cite{FV:10} we also get that $\|\widehat \X^2\|^\gamma_{1/p-\text{Höl};[0,T]}$ is finite almost surely and the claim follows.
\end{proof}

\begin{remark}\label{rem4}
Alternatively, an $L^2$-error for the stochastic Taylor expansion is given in Proposition~5.10.1 in \cite{KP:92}. It provides (under some technical conditions, in particular second moments of $Y_t$) the following estimate 
$$
\E\bigg[\Big|Y_t-Y_0
-\sum_{0<|I|\leq m}
\Psi_I(Y_0)\langle e_I,\widehat  \X_t\rangle\Big|^2\bigg]
\leq  Ct^{m+1},$$
for some $C\geq 0$. An application of the Markov inequality then yields the previous estimate
$$
\P\bigg(\Big|Y_t-Y_0
-\sum_{0<|I|\leq m}
\Psi_I(Y_0)\langle e_I,\widehat  \X_t\rangle\Big|>\widetilde C^{1/2} t^{(m+1)/2}\bigg)\leq \e,$$
where $\widetilde C:=C/\e$, for each $\e>0$.
\end{remark}

\section{The model} \label{sec:model}

We now introduce a framework for signature based asset price models. 
To this end we fix a time horizon $T>0$, consider a $d$-dimensional continuous semimartingale $\Xx:=(\Xx^1,\ldots, \Xx^{d})$ and its  matrix-valued quadratic covariation $[\Xx]$. We suppose that $\Xx$ encodes all the  information to represent the market's assets $\Ss$. For notational convenience we only consider a single asset, i.e.~assume that $\Ss$ is one-dimensional.

We shall furthermore assume that $\Xx$ has some tractability properties which are made precise below and which are for instance satisfied by a
$d$-dimensional Brownian motion (see Section~\ref{secBM} and  Example~\ref{ex3}). Alternatively, one can choose $\Xx$ to be a collection of liquid financial products, whose time series  are easily accessible and can be interpreted as realizations of continuous semimartingales.

\subsection{Definition and first properties}

As the continuous semimartingale $\Xx$ shall serve as main modeling building block, we call it \emph{primary process} and  extend it appropriately as made precise in the next definition.

\begin{definition}\label{def1}
We refer to $\Xx$ as \emph{primary process}. Depending on the context, we then consider  one of the following two extensions of $\Xx$.
\begin{enumerate}
\item\label{it2} $\widehat \Xx_t:=(t,\Xx_t,[\Xx]_t)$, where $[\Xx]$ denotes the $d^2$-dimensional process given by the quadratic covariation of $\Xx$.
\item\label{it1} $\widehat \Xx_t:=(t, \Xx_t)$. This extension is always be paired with the assumption that $\Xx$ is an It\^o-semimartingale with absolutely continuous characteristics such that each element of its diffusion matrix can be written as linear combination of elements of its time extended signature.
\end{enumerate}
\end{definition}

\begin{remark}
Note that \ref{it2} of Definition \ref{def1} is more general than \ref{it1} in two respects: first the components of the diffusion matrix do not need to be  linear functions of the time extended signature, but could for instance be more general path-dependent functionals; second $t \mapsto [X]_t$ does not need to be absolutely continuous with respect to the Lebesgue measure, hence $X$ does not need to be an It\^o-semimartingale.
\end{remark}

Throughout the paper we will always denote with $\widehat{\X}_{t}$ the signature of the previous extensions of $X_{t}$. Our goal consists in describing\slash approximating the dynamics of $\Ss$ with a \emph{signature model}.

\begin{definition} \label{def:model}
A \emph{signature model} is a stochastic process of the form
\begin{equation}\label{eqn3}
\Ss_n(\ell)_t:=\ell_{\emptyset}
+\sum_{0<|I|\leq n} \ell_I \langle e_I,\widehat \X_t\rangle,
\end{equation}
where $n\in\N$ and
$\ell:=\{\ell_\emptyset, \ell_I\colon 0<|I|\leq n\}$. 
\end{definition} 

\begin{remark}
By Proposition \ref{prop1} below the class of Sig-SDEs models considered in \cite{PSS:20} can be embedded in our framework by choosing a properly extended one-dimensional Brownian motion as primary process.
\end{remark}

In the following we list several important properties which make signature models a tractable framework for stochastic finance.
\begin{itemize}
    \item For each $t \in [0,T]$, $\Ss_n(\ell)_t$ is linear in $\widehat\X_t$. This in particular implies that having pre-computed the signature of  $\widehat\Xx$ an update of the parameters $\ell$ boils down to computing \eqref{eqn3}, which is nothing else than a scalar product.
 \item The quadratic variation of processes of  form \eqref{eqn3} is again of the form \eqref{eqn3}. 
 \item Form \eqref{eqn3} remains invariant under polynomial transformations.
 
 \item It\^o-integrals of processes of form \eqref{eqn3} with respect to processes of form \eqref{eqn3} are again processes of form \eqref{eqn3}. This includes in particular the signature $\widehat\S_n(\ell)$ of $\widehat\Ss_n(\ell)_t:=(t,\Ss_n(\ell)_t)$ or expressions of the form
 $$\int_0^\cdot S_n(
 \ell)_s \d\Xx_s^i.$$
 \item The latter implies that the expected signature of $\widehat\Ss_n(\ell)$ is given by
  $$\E[\langle e_J,\widehat\S_n(\ell)_t\rangle]=P_J(\ell,\E[\widehat\X_t]),$$
 for some  $P_J$ such that $P_J(\fdot, \E[\widehat\X_t])$ is a polynomial of degree $|J|$ and $P_J(\ell, \fdot)$ is a linear map for each $\ell$ (see Theorem \ref{thm1} and Remark \ref{rem2} for more details).
 \item Due to Theorem~\ref{universality} this provides approximations  
 $$\E\Big[f\Big((\widehat\S^2_n(\ell)_t)_{t\in[0,T]}\Big)\Big]\approx P_f(\ell,\E[\widehat\X_T])$$
 for each map $f$, which is continuous with respect to $d_{\Scal^{(2)}}$, where $P_f$ is given by a finite linear combination of maps $P_J$ as above. This includes representations for
$$\E\Big[\tilde f(\widehat\Ss_n(\ell)_T)\Big]\text{\qquad and \qquad}
\E\bigg[\tilde f\bigg(\int_0^T\widehat\Ss_n(\ell)_t \d t\bigg)\bigg]$$
for  maps $\tilde f$ being payoff functions and where the expectation is taken with respect to a pricing measure.
\end{itemize}

We now prove several representation results for the signature model of form \eqref{eqn3}. Throughout  we shall apply the following notation.

\begin{notation}\label{rem1}
We use $e_0$ for the component of $\widehat\Xx$ corresponding to time, $e_k$ for its component corresponding to $\Xx^k$, and (if needed) $\e_{ij}$ for the component of $\widehat\Xx$ corresponding to $[\Xx^i,\Xx^j]$.
\end{notation}

Consider the following assumption, which includes in particular the case where $\Xx$ is a $d$-dimensional Brownian motion.
\begin{assumption}\label{ass1}
For all $i,j\in\{1,\ldots, d\}$ it holds 
$$\d[\Xx^i,\Xx^j]_t=\sum_{ |I|\leq m} a_{ij}^I\langle e_I,\widehat \X_t\rangle \d t$$ for some $m\in\N$ where $\widehat\Xx_t=(t,X_t)$. 
 \end{assumption}
 
 \begin{remark}
 Note that for $\widehat{X}$ as of  Definition~\ref{def1}\ref{it1}, Assumption~\ref{ass1} is satisfied by definition. 
 \end{remark}
 
 \begin{proposition}\label{prop2}
Suppose that  
$\Ss_n(\ell)$ has a representation of form \eqref{eqn3} with $\widehat \Xx$ as in Definition~\ref{def1}\ref{it2} and that Assumption~\ref{ass1} holds. Then $\Ss_n(\ell)$ 
 has a representation of the same form but with $\widehat \Xx$ as in~Definition~\ref{def1}\ref{it1}.
 \end{proposition}
 \begin{proof}
 Observe that by Assumption~\ref{ass1} it holds
 $$[\Xx^i,\Xx^j]_t=
 \int_0^t \sum_{ |I|\leq m} a_{ij}^I\langle e_I,\widehat \X_s\rangle \circ \d s=
 \sum_{ |I|\leq m} a_{ij}^I\langle e_I\otimes e_0,\widehat \X_t\rangle.$$
The claim follows by Lemma~\ref{lem1}.
 \end{proof}
 \begin{proposition}\label{prop1}
 Fix $n\in \N$ and suppose that $\widehat \Xx$ is given by  Definition~\ref{def1}\ref{it2}. Then 
 there is a one to one correspondence between representations of the form \eqref{eqn3} and representations of the form
\begin{align*}
&\Ss_n(\ell)_t=\ell_{\emptyset}+\int_0^t\Big(\ell^0_\emptyset+\sum_{0<|I|\leq n-1} \ell^0_I \langle e_I,\widehat \X_s\rangle\Big) \d s
+\sum_{k=1}^{d}\int_0^t\Big(\ell^{k}_\emptyset+\sum_{0<|I|\leq n-1} \ell^{k}_I \langle e_I,\widehat \X_s\rangle \Big)\d \Xx_s^k\\
&\qquad+\sum_{k_1,k_2=1}^{d}\int_0^t\Big(\ell^{k_1,k_2}_\emptyset+\sum_{0<|I|\leq n-1} \ell^{k_1,k_2}_I \langle e_I,\widehat \X_s\rangle \Big)\d [\Xx^{k_1},\Xx^{k_2}]_s,
\end{align*}
for
$\ell:=\{\ell_\emptyset, \ell^0_I,\ell_I^{k_1}, \ell_I^{k_1,k_2}\colon  |I|\leq n-1\text{ and } k_1,k_2\in\{1,\ldots,d\}\}.$

If $\widehat \Xx$ satisfies Assumption~\ref{ass1}, then the same is true with $\ell_I^{k_1,k_2}=0$ for each $ |I|\leq n-1\text{ and } k_1,k_2\in\{1,\ldots,d\}$.
 \end{proposition}
 Before presenting the proof of this proposition let us formulate the following technical lemma.

 \begin{lemma}\label{lem2}
Suppose that $\widehat \Xx$ is given by Definition~\ref{def1}\ref{it2}. Then 
\begin{align*}
\int_0^t \langle e_I,\widehat \X_s\rangle \d s
&= \langle e_I\otimes e_0,\widehat \X_t\rangle,\qquad
\int_0^t \langle e_I,\widehat \X_s\rangle \d \Xx_s^k= \langle \tilde e_I^k,\widehat \X_s\rangle, \text{\quad and\quad}\\
\int_0^t \langle e_I,\widehat \X_s\rangle \d [\Xx^{k_1},\Xx^{k_2}]_s
&= \langle e_I\otimes \e_{k_1k_2},\widehat \X_t\rangle,
\end{align*}
where $\tilde e_\emptyset^k=e_k$ and 
$\tilde e_I^k=e_I\otimes e_{k}-\frac 1 2 e_{I'}\otimes \e_{i_{|I|}k} $ for each $|I|>0$. If Assumption~\ref{ass1} is in force, then the same holds true with
$$\tilde e_I^k=e_I\otimes e_{k}- \sum_{|J|\leq m}\frac {a_{i_{|I|}k}^J} 2(e_{I'}\shuffle e_J)\otimes e_0 $$
 for each $|I|>0$. 
 \end{lemma}

 \begin{proof}
 The representations of $\int_0^t \langle e_I,\widehat \X_s\rangle \d s$ and $\int_0^t \langle e_I,\widehat \X_s\rangle \d [\Xx^{k_1},\Xx^{k_2}]_s$ follow by definition of the signature.
 We proceed with the proof of the representation of $\int_0^t \langle e_I,\widehat \X_s\rangle \d \Xx_s^k$. For $I=\emptyset$ the claim is clear.
 By definition of the Stratonovich integral we have that
\begin{align*}
\int_0^t\langle e_{I},\widehat \X_s\rangle \d \Xx^k_s
&=\int_0^t\langle e_{I},\widehat \X_s\rangle \circ \d X^k_s
-\frac 1 2   [\langle e_{I},\widehat \X\rangle,X^k]_t\\
&=\langle e_{I}\otimes e_k,\widehat \X_t\rangle
-\frac 1 2   \int_0^t\langle e_{I'},\widehat \X_s\rangle \d [X^{i_{|I|}},X^k]_s\\
&=\langle e_{I}\otimes e_k,\widehat \X_t\rangle
-\frac 1 2   \langle e_{I'}\otimes \e_{i_{|I|}k},\widehat \X_t\rangle\\
&=\langle \tilde e_I^k,\widehat \X_t\rangle.
\end{align*}
Suppose now that Assumption~\ref{ass1} is in force. Then Proposition~\ref{shuffle-property} and the definition of the signature yield
\begin{align*}
\int_0^t \langle e_I,\widehat \X_s\rangle \d [\Xx^{k_1},\Xx^{k_2}]_s
&=\int_0^t \sum_{|J|\leq m}{a_{k_1k_2}^J}  \langle e_I,\widehat \X_s\rangle \langle e_J,\widehat\X_s\rangle \d s\\
&=\sum_{|J|\leq m} {a_{k_1k_2}^J} \langle(e_I\shuffle e_J)\otimes e_0,\widehat\X_t\rangle,
\end{align*}
and the claim follows.
\end{proof}

We are now ready to provide the proof of Proposition~\ref{prop1}.
\begin{proof}[Proof of Proposition~\ref{prop1}]
Let $\Ss_n(\ell)$ be as in the statement of the proposition. By Lemma~\ref{lem2} it holds 
\begin{align*}
&\Ss_n(\ell)_t=\ell_{\emptyset}+\Big(\ell^0_\emptyset\langle e_0,\widehat \X_t\rangle+\sum_{0<|I|\leq n-1} \ell^0_I \langle e_I\otimes e_0,\widehat \X_t\rangle\Big) 
+\sum_{k=1}^{d}\Big(\ell^{k}_\emptyset\langle \tilde e ^k_\emptyset,\widehat \X_t\rangle+\sum_{0<|I|\leq n-1} \ell^{k}_I \langle \tilde e^k_I,\widehat \X_t\rangle \Big)\\
&\qquad+\sum_{k_1,k_2=1}^{d}\Big(\ell^{k_1,k_2}_\emptyset\langle \e_{k_1k_2},\widehat \X_t\rangle+\sum_{0<|I|\leq n-1} \ell^{k_1,k_2}_I \langle e_I\otimes\e_{k_1k_2},\widehat \X_s\rangle \Big),
\end{align*}
which is of the form given by \eqref{eqn3}.
Conversely, by Lemma~\ref{lem2} we also have
\begin{align*}
\langle e_I\otimes e_0,\widehat \X_t\rangle
&=\int_0^t \langle e_I,\widehat \X_s\rangle \d s,\qquad
\langle e_I\otimes \e_{k_1k_2},\widehat \X_t\rangle=\int_0^t \langle e_I,\widehat \X_s\rangle \d [\Xx^{k_1},\Xx^{k_2}]_s, \text{\quad and\quad}\\
\langle  e_I^k\otimes e_k,\widehat \X_s\rangle
&=\int_0^t \langle e_I,\widehat \X_s\rangle \d \Xx_s^k
+\frac 1 2 \int_0^t  \langle e_{I'},\widehat \X_s\rangle\d [\Xx^{i_{|I|}},\Xx^{k}]_s,
\end{align*}
and the claim follows.
\end{proof}

In the case of Brownian motion the formulas of Lemma~\ref{lem2} simplify as follows.

\begin{example}\label{lemma:tilde:multi}
Suppose that $\Xx$ is a vector of correlated Brownian motions with correlation matrix $\rho$ and $\widehat \Xx$ is given by Definition~\ref{def1}\ref{it1}. Then it holds
$$\tilde e^k_\emptyset=e_{k}\qquad\text{and}\qquad\tilde e^k_I=e_{I}\otimes e_{k}
-\frac {\rho_{i_{|I|},k}} 2  1_{\{i_{|I|}\neq 0\}}e_{I'}\otimes e_0.$$
\end{example}

\subsection{Absence of arbitrage and universality properties}\label{sec1}

The above representation results will allow us to 
establish conditions that guarantee absence of arbitrage. We shall therefore assume that \emph{no free lunch with vanishing risk} (\cite{DS:94}) holds, which is -- due to the continuity of sample paths of $\Ss_{n}(\ell)$ -- equivalent to the existence of an equivalent local martingale measure $\mathbb{Q}$, i.e.~a measure $\mathbb{Q} \sim \mathbb{P}$ under which the (discounted) asset price model $S_{n}(\ell)$ is a local martingale. We shall always assume here that interest rates are zero and that the asset is already discounted.
To formulate a precise no-arbitrage condition, the following corollary  which is 
 a direct consequence of Proposition~\ref{prop1} and Lemma \ref{lem2} is essential.
 
\begin{corollary}\label{cor1}
Suppose that $\Xx$ is a local martingale. Then $\Ss_{n}(\ell)$ is a local martingale if and only if it admits a representation of the form
\begin{equation}\label{eq:martingality}
\Ss_{n}(\ell)_t=\ell_{\emptyset}
+\sum_{k=1}^{D}\bigg(\ell_\emptyset^k\langle e_k, \widehat\X_t\rangle+\sum_{0<|I|\leq n-1} \ell^{k}_I \langle \tilde e^k_I,\widehat \X_t\rangle\bigg),
\end{equation}
for some $D\in\{1,\ldots, d\}$ and $\ell:=\{\ell_\emptyset, \ell_I^{k}\colon 0\leq|I|\leq n-1\text{ and } k\in\{1,\ldots,d\}\}.$
\end{corollary}

\begin{proof}
Due to the local martingale property of $X$, $S_n(\ell)$ in the representation of Proposition~\ref{prop1} is a local martingale if and only if all integrals with respect to time and the quadratic variation process vanish. This means that $S_n(\ell)$ is of form
\begin{align}\label{eq:martingalecond}
S_n(\ell)=\ell_{\emptyset}+\sum_{k=1}^{d}\int_0^t\Big(\ell^{k}_\emptyset+\sum_{0<|I|\leq n-1} \ell^{k}_I \langle e_I,\widehat \X_s\rangle \Big)\d \Xx_s^k
\end{align}
and Lemma  \ref{lem2}  yields the assertion.
\end{proof}

We are now ready to formulate sufficient no-arbitrage conditions.

\begin{corollary}

Suppose that  there is an equivalent measure  $\mathbb{Q} \sim \mathbb{P}$ such that $X$ is a local $\mathbb{Q}$-martingale. Then the following holds.
\begin{enumerate}
\item 
The model  $S_{n}(\ell)$  is free of arbitrage if it admits a representation as of \eqref{eq:martingality}.
\item 
If $\mathbb{Q}$ is  an equivalent local martingale measure 
for $S_{n}(\ell)$,
then $S_{n}(\ell)$ is necessarily of form \eqref{eq:martingality}.

\end{enumerate}

\end{corollary}

\begin{proof}
The first assertion is a direct consequence of  Corollary \ref{cor1}, as form \eqref{eq:martingality} implies that $S_{n}(\ell)$ is a local $\mathbb{Q}$-martingale and thus $\mathbb{Q}$ is an equivalent local martingale measure.
The assumption of the second assertion
implies that $S_{n}(\ell)$ is a local martingale under $\mathbb{Q}$, whence by Corollary \ref{cor1} it has to be of form \eqref{eq:martingality}.
\end{proof}

\begin{remark}
\begin{enumerate}
\item 
Note that $S_{n}(\ell)$ could be free of arbitrage without being of form \eqref{eq:martingality}. This is the case if
the primary process $X$ is not a local martingale under any 
of the equivalent local martingale measures.

\item 
Observe also that if $S_n(\ell)$ is form \eqref{eq:martingality} and a local martingale under $\mathbb{Q}$, then $X$ does not necessarily need to be a local $\mathbb{Q}$-martingale.
This happens if  drift terms  in \eqref{eq:martingalecond} cancel out. 
Note however that in dimension $d=1$ this cannot occur and $X$ is thus necessarily a local $\mathbb{Q}$-martingale. In particular, in this case $\mathbb{Q}$ is unique and the model thus complete.
\end{enumerate}
\end{remark}

In the following example we show how classical stochastic volatilty models can be approximated by arbitrage-free signature models, thus making the announced universality properties precise.

\begin{example}\label{example:tilde:sv:1d}
We consider a generic stochastic volatility model driven by two correlated Brownian motions $(B_{t})_{t\in[0,1]}$ and $(W_{t})_{t\in[0,1]}$ with correlation coefficient $\rho\in[-1,1]$. More precisely, we assume that under some equivalent local martingale measure $\Q$
the dynamics of $S$  are given by
\begin{align*}
	{\d}S_{t}&=g(S_t,V_t){\d}B_{t}, \label{eq:stochvol}\\
	{\d}V_{t}&=h(S_t,V_t){\d}t+\sigma(S_t,V_{t}){\d}W_{t}, \notag
\end{align*} 
for some functions $g,h,\sigma:\mathbb{R}^2\to\mathbb{R}$. 
As a possible choice for the parameters we can for instance consider
\begin{itemize}
	\item  $g(s,v)=s\sqrt{v}$, $h(s,v)=\kappa(\theta-v)$ and $\sigma(s,v)=\sigma_{0}\sqrt{v}$, for some parameters $\kappa, \theta \geq 0$ and $\sigma_0 \in \mathbb{R}$ to retrieve the Heston model (see \cite{H:93});
	\item $g(s,v)=s^{\beta}v$, $h(s,v)=0$ and $\sigma(s,v)=\alpha v$ for $0 < \beta \leq 1 $ and $\alpha \in \mathbb{R}$ to retrieve the SABR model (see \cite{HKLW:02}).
\end{itemize}
Set now  $X_{t}:=(B_{t},W_{t})$ and let $\widehat X_{t}$ be given by the representation of Definition~\ref{def1}\ref{it1}, i.e., $\widehat{X}_t:=(t, B_{t},W_{t})$. For some fixed $n \in \mathbb{N}$, consider then the following model 
\begin{equation*}
\Ss_{n}(\ell)_t:=\ell_\emptyset+\int_{0}^{t}\Big(\ell_{\emptyset}^B+\sum_{0<\lvert I\lvert\le n-1}\ell_{I}^B\langle e_{I}^B,  \widehat\X_{s}\rangle\Big){\d}B_{s},
\end{equation*}
 which is  a local martingale under $\Q$
and thus absence of arbitrage is guaranteed. 
Moreover, by Proposition~\ref{prop1}, we know that this model has a representation of the form 
\begin{align*}
    \Ss_{n}(\ell)_t=\ell_\emptyset+\ell_\emptyset^B B_t+\sum_{0<|I|\leq n-1} \ell_I^B \langle \tilde e_I^B,\widehat\X_t\rangle,
\end{align*}
for each $t\in[0,1]$, where by Lemma~\ref{lem2}
\begin{align}\label{eq:SnHestSabr}
	\tilde e_I^{B}&:=e_{I}\otimes e_1 
	-\frac{1}{2}(1_{\{i_{|I|}=1\}}+\rho 1_{\{i_{|I|}=2\}})e_{I'}\otimes e_0, 
\end{align}
for each $I$ with $0<|I|\leq n-1$, implying that it is a signature model as of Definition \ref{def:model}.

Concerning universality note that if
the coefficients $\ell_I^B$ are chosen
such that \eqref{eq:SnHestSabr}
matches the signature approximation of $S$ up to order $n-1$  as specified in Proposition~\ref{prop214} (assuming appropriate regularity conditions on the coefficients $g, h$, $\sigma$), then at least locally in $t$, $S_n(\ell)$ is close to $S$ in probability.

Let us exemplify this by means of the SABR model for $\beta=1$ and, for simplicity, $\rho=0$. Set $\ell_\emptyset:=S_0$, 
$\ell_\emptyset^B:=y_1y_2$, 
$\ell_{(0)}^B:=-\frac 1 2 y_1y_2^3-\frac 1 2 \alpha^2y_1y_2$,
$\ell_{(1)}^B:= y_1y_2^2$, and
$\ell_{(2)}^B:=\alpha y_1y_2$.
By an application of Proposition~\ref{prop214} for $X_t=(t,W_t,B_t)$, $Y^1=S$, $Y^2=V$, $m=2$,
$$\Psi_0(y)=(1,-\frac 1 2y_1y_2^2,-\frac 1 2 \alpha^2y_2)^\top,$$
$\Psi_1(y)=(0,y_1y_2,0)^\top$, and $\Psi_2(y)=(0,0,\alpha y_2)^\top$
we can conclude that for each $\e>0$ there is a constant  $C(\e)$ such that 
$$\Q\Big(|S_t
-S_2(\ell)_t|
> C(\e)t^{3/2} \Big)<\e,$$
where $S_2(\ell)$ is given by Definition~\ref{def:model} for $\widehat X_t=(t,B_t,W_t)$. 
The same procedure for general $n$ leads to a similar result for an arbitrary speed of convergence.

Another possibility to establish universality is to use the fact that the solution map of a stochastic differential equation (with sufficiently regular coefficients) is a continuous map of the signature of the driving signal (see e.g.~Corollary 10.40 \cite{FV:10}).
The universal approximation theorem (Theorem~\ref{universality}, see also the formulation of Lemma 5.2 in~\cite{BHRS:21} or \cite{CM:22}) then yields the result, however without quantitative estimates.
\end{example}

\subsection{The expected signature of \texorpdfstring{$\Ss_n(\ell)$}{S(l)}}\label{sec32}

For pricing purposes of so-called \emph{sig-payoffs} treated in Section \ref{sig-payoffs} it will be  important to be able to compute the expected signature of $\widehat S_n(\ell)_t=(t,S_n(\ell)_{t})$.  In this section we thus provide formulas which trace this computation back to the calculation of the expected signature of $\widehat{X}_t$, which in many cases is well-known (see e.g.~\cite{FW:03} for Brownian motion) and can often be computed by techniques of polynomial processes (see \cite{CST:22}) or by solving an infinite dimensional system of linear PDEs corresponding to the Kolmogorov forward equation of the signature process (see \cite{N:12}). For a unified treatment of signature cumulants, i.e.~the logarithm of expected signature, we refer to \cite{FHT:21}.

Here, we first suppose  that $\widehat \Xx$ is given by  Definition~\ref{def1}\ref{it1}.

\begin{theorem}\label{thm1}
Fix $n\in \N$, a multi-index $J$, and  $D\in\{1,\ldots, d\}$ and denote by $\widehat \S_n(\ell)_{t}$ the signature of  $\widehat S_n(\ell)_{t}$. Let $e_0$ be the component of  $\widehat S_n(\ell)$ corresponding to time and $e_1$ its component corresponding to $S_n(\ell)$. Define $e(\emptyset,\ell):=\tilde e(\emptyset,\ell):=e_\emptyset$ and 
\begin{align*}
e(J,\ell)&=\halfshuffle_{i=1}^{|J|}\Big(e_01_{\{j_{i}=0\}}+\Big(\sum_{0<\lvert I\lvert\le n}\ell_{I}e_{I}\Big)1_{\{j_{i}=1\}}\Big),\\
\tilde e(J,\ell)&=\halfshuffle_{i=1}^{|J|}\Big( e_01_{\{j_{i}=0\}}+\Big(\sum_{k=1}^D\Big(\ell_\emptyset^k e_k+\sum_{0<\lvert I\lvert\le n-1}\ell_{I}^k\tilde e^k_{I}\Big)\Big)1_{\{j_{i}=1\}}\Big),
\end{align*}
for $|J|>0$ with $\halfshuffle$ the half-shuffle being introduced in Definition \ref{def:halfshuffle}.
Then the following representations hold.
\begin{itemize}
\item $
		\langle e_{J},{\widehat\S_n(\ell)}_{t}\rangle
		=\langle e(J,\ell),\widehat{\X}_{t}\rangle$, if $S_n(\ell)$ is given by \eqref{eqn3}, and
		\item		$\langle  e_{J},{\widehat\S_n(\ell)}_{t}\rangle
		=\langle \tilde e(J,\ell),\widehat{\X}_{t}\rangle$, if $\Ss_{n}(\ell)$ is given as in Corollary~\ref{cor1}.
\end{itemize}

\end{theorem}
\begin{proof}
	In order to prove the claim we proceed by induction. Fix $S_n(\ell)$ as in \eqref{eqn3}. For $J=\emptyset$ the claim follows by the definition of signature. Suppose the claim holds true for each $J$ such that $|J|=m-1$, and fix $J$ with $|J|\leq m$. Then
	\begin{align*}
		\langle e_{J}, \widehat\S_n(\ell)_t\rangle&=\int_{0}^{t}\langle e_{J'},\widehat\S_n(\ell)_{s}\rangle\circ {\d}\langle e_{j_{m}}, \widehat\Ss_n(\ell)_s\rangle\\
		&=\int_{0}^{t}\langle e_{J'},\widehat\S_n(\ell)_{s}\rangle\circ{\d}\langle e_01_{\{j_m=0\}}
		+\Big(\sum_{0<\lvert I\lvert\le n}\ell_{I}e_{I}\Big)1_{\{j_m=1\}},\widehat\X_s\rangle\\
		&=1_{\{j_m=0\}}\int_{0}^{t}\langle e_{J'},\widehat\S_n(\ell)_{s}\rangle\circ{\d}\langle e_0,\widehat\X_s\rangle\\
		&\qquad+1_{\{j_m=1\}}\sum_{0<\lvert I\lvert\le n}\ell_{I}\int_{0}^{t}\langle e_{J'},\widehat\S_n(\ell)_{s}\rangle\circ{\d}
		\langle e_{I},\widehat\X_s\rangle.
	\end{align*}
	The induction hypothesis and
Lemma~\ref{lem1} yield the first claim and the second one is analogous.
\end{proof}

\begin{remark}\label{rem2}\phantomsection
\begin{enumerate}
    \item
Let now 
$ \Ss_n(\ell)$ be as in \eqref{eqn3}
and observe that 
$$e(J,\ell)
=
\halfshuffle_{i=1}^{|J|}\sum_{0<|I|\leq n}\Big(e_I1_{\{j_i=0\}}1_{\{I=(0)\}}+\ell_{I}e_{I}1_{\{j_{i}=1\}}\Big).$$
Setting $c(j,I,\ell):=1_{\{j=0\}}1_{\{I=(0)\}}+\ell_{I}1_{\{j=1\}}$ we thus obtain that
\begin{align*}
\E[\langle e_{J},{\widehat\S_n(\ell)}_{t}\rangle]
&=\E[\langle e(J,\ell),\widehat{\X}_{t}\rangle]
=\sum_{|I_1|,\ldots,|I_{|J|}|=1}^n\E[\langle\halfshuffle_{i=1}^{{|J|}}e_{I_i},\widehat{\X}_{t}\rangle]\prod_{i=1}^{|J|}c(j_i,I_i,\ell).
\end{align*}
Even if at a first sight this representation  appears involved, it is in fact very handy. Indeed the expectations $\E[\langle\halfshuffle_{i=1}^{{|J|}}e_{I_i},\widehat{\X}_{t}\rangle]$ can be computed just once in advance. Since $c(j,I,\fdot)$ is affine, we also immediately obtain that the map
$$(\ell,\E[\widehat\X])\mapsto P_J(\ell,\E[\widehat \X_t]):=\E[\langle e_{J},{\widehat\S_n(\ell)}_{t}\rangle]$$
is polynomial of degree $|J|$ in its first argument and linear in the second one.

\item Similarly, let  
$ \Ss_{n}(\ell)$  be as in Corollary~\ref{cor1},
set $\tilde e_\emptyset^k:=e_k$,
and observe that
$$
\tilde e(J,\ell)=\halfshuffle_{i=1}^{|J|}\sum_{k=1}^D\sum_{0\leq |I|\leq n-1}
\Big( \tilde e_I^k1_{\{j_{i}=0\}}1_{\{I=(0)\}}1_{\{k=1\}}+\ell_{I}^k\tilde e^k_{I}1_{\{j_{i}=1\}}\Big).
	$$
Setting $\tilde c(j,I,k,\ell):=1_{\{j=0\}}1_{\{I=(0)\}}1_{\{k=1\}}+\ell_{I}^k1_{\{j=1\}}$ yields
\begin{align*}
\E[\langle e_{J},{\widehat\S_{n}(\ell)}_{t}\rangle]=\sum_{k_1,\ldots,k_{|J|}=1}^D\sum_{|I_1|,\ldots,|I_{|J|}|=0}^{n-1}\E[\langle\halfshuffle_{i=1}^{{|J|}}\tilde e_{I_i}^{k_i},\widehat{\X}_{t}\rangle]
\prod_{i=1}^{|J|}\tilde c(j_i,I_i,k_i,\ell).
\end{align*}
We again obtain that the map
$$(\ell,\E[\widehat\X])\mapsto \widetilde P_J(\ell,\E[\widehat \X_t]):=\E[\langle e_{J},{\widehat\S_{n}(\ell)}_{t}\rangle]$$
is polynomial of degree $|J|$ in its first argument and linear in the second one.

\end{enumerate}
\end{remark}
The expressions above clearly also provide a formula of the variance of $\Ss_{n}(\ell)$.

\begin{corollary}Let
$ \Ss_{n}(\ell)$  be as in Corollary~\ref{cor1} and assume that it is a true martingale. Then
		\begin{equation}
			\Var(S_{n}(\ell)_{t})=2\sum_{k_1,k_2=1}^D\sum_{|I_1|,|I_{2}|=1}^{n-1}\mathbb{E}[\langle \tilde e_{I_{1}}^{k_1}\shuffle \tilde e_{I_2}^{k_2},\widehat{\mathbb{X}}_{t}\rangle]\ell_{I_{1}}\ell_{I_{2}}.
		\end{equation}
\end{corollary}
\begin{proof} The martingale property guarantees that $\E[S_{n}(\ell)_{t}]=S_{n}(\ell)_0$ and hence, by Proposition~\ref{shuffle-property} (see also Example~\ref{ex2}), $\Var (S_{n}(\ell)_{t})=2\mathbb{E}[\langle e_{1}\otimes e_{1},\widehat{\S}_{n}(\ell)_{t}\rangle]$. By Remark~\ref{rem2} we can conclude that 
$$\Var (S_{n}(\ell)_{t})=2\sum_{k_1,k_2=1}^D\sum_{|I_1|,|I_{2}|=1}^{n-1}\mathbb{E}[\langle \tilde e_{I_{1}}^{k_1}\shuffle \tilde e_{I_2}^{k_2},\widehat{\mathbb{X}}_{t}\rangle]\tilde c(1,I_1,k_1,\ell)\tilde c(1,I_2,k_2,\ell)$$ and the claim follows.
\end{proof}

Let us now consider the case where $X$  is additionally extended by its quadratic variation. In that case we need the following representation.
\begin{lemma}\label{lem4}
Suppose that $\widehat \Xx$ is given by  Definition~\ref{def1}\ref{it2} and fix two multi-indices $I$ and $J$ such that $|I|>0$ and $|J|>0$. 
Using the notation of Remark~\ref{rem1} set 
$$
	e_{I}[ \shuffle] e_{J}:=\sum_{k_1,k_2=1}^{d}1_{\{i_{|I|}=k_1, j_{|J|}=k_2\}}( e_{I'}\shuffle e_{J'})\otimes \e_{{k_1k_2}}.
$$
Then
$[\langle e_{\emptyset},\widehat\X\rangle,\langle e_{J},\widehat\X\rangle]_t
=[\langle e_{\emptyset},\widehat\X\rangle]_t
=0$, and 
$[\langle e_{I},\widehat\X\rangle,\langle e_{J},\widehat\X\rangle]_t
=\langle e_I[\shuffle]e_J,\widehat \X_t\rangle$. If Assumption~\ref{ass1} is in force the same is true with
$$
	e_{I}[ \shuffle] e_{J}:=\sum_{k_1,k_2=1}^{d}\sum_{|H|\leq m} {a_{k_1k_2}^H} 1_{\{i_{|I|}=k_1, j_{|J|}=k_2\}}(e_{I'}\shuffle e_{J'}\shuffle e_H)\otimes e_0.
$$
\end{lemma}
\begin{proof}
By definition of the signature and Proposition~\ref{shuffle-property} we can compute
\begin{align*}
[\langle e_{I},\widehat\X\rangle,\langle e_{J},\widehat\X\rangle]_t&=\int_0^t \langle e_{I'},\widehat \X_s\rangle 
\langle e_{J'},\widehat \X_s\rangle 
\d[\langle e_{i_{|I|}},\widehat \X\rangle ,\langle e_{j_{|J|}},\widehat \X\rangle]_s\\
&=\sum_{k_1,k_2=1}^{d}1_{\{i_{|I|}=k_1,j_{|J|}=k_2\}}
\langle e_{I'},\widehat \X_s\rangle 
\langle e_{J'},\widehat \X_s\rangle 
\d\langle \e_{k_1k_2},\widehat \X_s\rangle\\
&=\langle e_I[\shuffle]e_J,\widehat \X_t\rangle.
\end{align*}
Since under Assumption~\ref{ass1} it holds
$\d[\langle e_{i_{|I|}},\widehat \X\rangle ,\langle e_{j_{|J|}},\widehat \X\rangle]_t=\sum_{|H|\leq m} {a_{k_1k_2}^H} \langle e_H,\widehat\X_t\rangle \d\langle e_0,\widehat\X_t\rangle,$
the second claim follows again by Proposition~\ref{shuffle-property}.
\end{proof}

Using the above lemma we can now express the signature components of $\widehat S_n(\ell)_{t}$ here defined as $\widehat S_n(\ell)_{t}:=(t,S_n(\ell)_{t},[S_n(\ell)]_t)$ also via $\widehat{\X}_{t}$.

\begin{theorem}\label{exp-sig-X-lemma}
Fix $n\in \N$, a multi-index   $J$, and suppose that $\widehat \Xx$ is given by  Definition~\ref{def1}\ref{it2}. Set $\widehat S_n(\ell)_{t}:=(t,S_n(\ell)_{t},[S_n(\ell)]_t)$ and let $\widehat \S_n(\ell)_{t}$ denote the corresponding signature. Let $e_0$ be the component of  $\widehat S_n(\ell)_{t}$ corresponding to time, $e_1$ its component corresponding to $S_n(\ell)_{t}$, and $e_2$ its component corresponding to $[S_n(\ell)]_{t}$. Define $e(J,\ell):=e_\emptyset$ for $J=\emptyset$ and
$$e(J,\ell)=\halfshuffle_{k=1}^{|J|}\Big(e_01_{\{j_{k}=0\}}+\Big(\sum_{0<\lvert I\lvert\le n}\ell_{I}e_{I}\Big)1_{\{j_{k}=1\}}
+\Big(\sum_{0<|I|,|H|\leq n}\ell_{I}\ell_{H} e_{I}[\shuffle]e_{H}\Big)1_{\{j_{k}=2\}}
\Big).$$
Then
$
		\langle e_{J},{\widehat\S_n(\ell)}_{t}\rangle
		=\langle e(J,\ell),\widehat{\X}_{t}\rangle.$
		\end{theorem}
\begin{proof}
First observe that by Lemma~\ref{lem4} it holds  
$$[\Ss_n(\ell)]=\Big[\Big(\ell_\emptyset+\sum_{0<\lvert I\lvert\le n}\ell_{I}\langle e_{I},\widehat\X\rangle\Big)\Big]
=\sum_{0<|I|,|H|\leq n}\ell_{I}\ell_{H}\langle e_{I}[\shuffle]e_{H},\widehat \X\rangle.$$
Next, note that
	\begin{align*}
		\langle e_{j},{\widehat\Ss_n(\ell)}_{t}\rangle
		&=t1_{\{j=0\}}+ ({{\Ss_n(\ell)}}_{t}-{{\Ss_n(\ell)}}_{0})1_{\{j=1\}}+[\Ss_n(\ell)]_t1_{\{j=2\}}\\
		&=\langle e_01_{\{j=0\}}+\Big(\sum_{0<\lvert I\lvert\le n}\ell_{I}e_{I}\Big)1_{\{j=1\}}
+\Big(\sum_{0<|I|,|H|\leq n}\ell_{I}\ell_{H} e_{I}[\shuffle]e_{H}\Big)1_{\{j=2\}},\widehat{\X}_{t}\rangle.
	\end{align*}
	The proof now follows the proof of Theorem~\ref{thm1}.
	\end{proof}

\subsection{Pricing of sig-payoffs}\label{sig-payoffs}
We recall here the notion of   \emph{sig-payoffs} as introduced in \cite{LNP:20}.
\begin{definition}
Suppose that the price process $S$ is given by a continuous semimartingale. A payoff $F:\Omega\to\mathbb{R}$ is said to be a sig-payoff if there exists $m\in\N$, and 
$f:=\{f_\emptyset, f_J\colon 0<|J|\leq m\},$ such that
\begin{align*}
    F:=f_\emptyset +\sum_{0<|J|\leq m}f_J\langle e_J, \widehat{\mathbb{S}}_{T}\rangle,
\end{align*}
where $\widehat \S$ denotes the signature of $\widehat S_t=(t,S_t)$.
\end{definition}

\begin{example}
Let $K>0$ be a strike price and $T>0$ a maturity time. Then, Asian forwards written on a stock $S$ are payoffs of the form
$$\frac{1}{T}\int_0^T S_t \d t -K =\frac{1}{T}\int_0^T (S_t -S_{0}) \d t -K+S_{0} =\frac{1}{T}\langle e_{1}\otimes e_{0},\widehat \S_T\rangle+(K-S_{0})\langle e_{\emptyset},\widehat \S_T\rangle,$$
and are thus sig-payoffs.
\end{example}

Even though the standard vanilla derivatives such as call and put options are \emph{not} sig-payoffs, approximate sig-payoffs can be used as efficient control variates in Monte Carlo pricing, which we outline in Section \ref{sec:MCvar}.

Note that in this section we shall use $\widehat S_t$ always for $\widehat S_t=(t,S_t)$ and consider  pricing of sig-payoffs when $S$ is given by a signature model, as made precise in the following corollary.

\begin{corollary}
Let the dynamics of $S_n(\ell)$ under a local martingale measure $\Q$ be specified as in Corollary~\ref{cor1} for $\widehat \Xx$  given by  Definition~\ref{def1}\ref{it1}. Consider a sig-payoff
$$
    F=f_\emptyset +\sum_{0<|J|\leq m}f_J\langle e_J, \widehat{\mathbb{S}}_{n}(\ell)_{T}\rangle.
$$
Then,  
using the notation of Remark~\ref{rem2} we can write the corresponding price as 
\begin{align}
    \E_\Q[F]&=f_\emptyset +\sum_{0<|J|\leq m}f_J \widetilde P_J(\ell,\E_\Q[\widehat\X_T]) \notag\\
    &=f_\emptyset +\sum_{|J|=1}^m\sum_{k_1,\ldots,k_{|J|}=1}^D\sum_{|I_1|,\ldots,|I_{|J|}|=0}^n\E[\langle\halfshuffle_{i=1}^{{|J|}}\tilde e_{I_i}^{k_i},\widehat{\X}_{T}\rangle]
    f_J
\prod_{i=1}^{|J|}\tilde c(j_i,I_i,k_i,\ell), \label{eq:sigpayoff}
\end{align}
for $\tilde c(j,I,k,\ell):=1_{\{j=0\}}1_{\{I=(0)\}}1_{\{k=1\}}+\ell_{I}^k1_{\{j=1\}}$.
\end{corollary}

\begin{proof}
This is a direct consequence of Theorem \ref{thm1} and Remark \ref{rem2}.
\end{proof}

\begin{remark}
\begin{enumerate}
\item
     Expression \eqref{eq:sigpayoff} admits also a second representation that turns out to be useful for coding:
\begin{align*}
    \E_\Q[F]
&=f_\emptyset+
\sum_{|J|=1}^m\sum_{k_1,\ldots,k_{|J|}=1}^D
\sum_{I_1,\ldots,I_{|J|}\in\Ical}\E[\langle\halfshuffle_{i=1}^{{|J|}}\tilde e_{I_i}^{k_i},\widehat{\X}_T\rangle]\\
&\qquad \times\, f_{(1_{\{I_{1}\neq I^t,k_1\neq 0\}},\ldots,1_{\{I_{{|J|}}\neq I^t,k_{{|J|}}\neq 0\}})}
\prod_{i=1}^{|J|}\Big(1_{\{I_{i}=I^t, k_i=0\}}+ \ell_{I_i}^k1_{\{I_{i}\neq I^t, k_i\neq 0\}}\Big).
\end{align*}
\item With a similar procedure it is also possible to provide a representation of $\Var_\Q(F)$. By Proposition~\ref{shuffle-property} we know that
$$(F-f_\emptyset)^2=\sum_{|J_1|,|J_2|=1}^mf_{J_1}f_{J_2}\langle e_{J_1}\shuffle e_{J_2}, \widehat{\mathbb{S}}_{n}(\ell)_{T}\rangle. $$
Setting 
$\widetilde P_{I_1\shuffle I_2}:=\sum_{i=1}^K \widetilde P_{J_i}$
for $J_i$ and $K$ satisfying $e_{I_1}\shuffle e_{I_2}=\sum_{i=1}^Ke_{J_i}$ we thus obtain
\begin{align*}
    \Var _\Q(F)
    &=\E_\Q[(F-f_\emptyset)^2]-\E_\Q[F-f_\emptyset]^2\\
    &=\sum_{|J_1|,|J_2|=1}^mf_{J_1}f_{J_2} \Big(\widetilde P_{J_1\shuffle J_2}(\ell,\E_\Q[\widehat \X_T])
    - \widetilde P_{J_1}(\ell,\E_\Q[\widehat\X_T])\widetilde P_{J_2}(\ell,\E_\Q[\widehat\X_T])\Big).
\end{align*}
\item From this representations we can see that the maps $P_F(f,\ell,\E[\widehat\X_T]):=\E[F]$ and $P_{\Var (F)}(f,\ell,\E[\widehat\X_T]):=\Var (F)$ inherit good properties from $\widetilde P_J$. In particular, $P_F$ is linear in $f$ and polynomial of degree $m$ in $\ell$, and $P_{\Var (F)}$ is quadratic in $f$ and polynomial of degree $2m$ in $\ell$. Both maps are linear in $\E[\widehat\X_T]$. 
\end{enumerate}
\end{remark}

\section{Calibration} \label{sec:calibration}
This section is dedicated to illustrate how the signature model can be calibrated to market data. We here consider three different tasks: first, calibration to time series data, second calibration to option prices, and third a combination of these two tasks.

 Throughout the section we fix $n \geq 1$ and $d\ge1$. Consider a model given as in Corollary~\ref{cor1} with $\widehat \Xx$ as in Definition~\ref{def1}\ref{it1}. In order to simplify the notation we set $\ell_\emptyset := S_0$ and we drop the index $1$ from $\ell^1_I$ and $\tilde e_I^1$.
 
\subsection{Calibration to time-series data}\label{sec:calibration_ts}

For the time-series calibration we investigate two different methods. In both cases we suppose that 
time series data for $\widehat X$ and $\Ss$ are available on a time grid $t_1,\ldots, t_N$, with $N>1$.

\subsubsection{Regression using price data}
Our first method consists in performing a \emph{linear} regression 
 where we directly regress -- according to the signature model as in Corollary~\ref{cor1} for $D=1$ --
 the trajectories of $S$ on the signature of $\widehat X$.
 For the computation of  $(\widehat\X_{t_i})_{i=1}^N$ based on the observations $(\widehat\Xx_{t_i})_{i=1}^N$  we can resort to several available  packages, e.g.~the package \verb|iisignature| developed by \cite{RG:18} or \verb|signatory| by \cite{KL:20} in Python.

Let  $d^\ast:=\frac{(d+1)^{n}-1}{d}$ denote the dimension of the signature vector of $\widehat{X}$ truncated at level $n-1$. Then the current calibration problem consists in finding $\ell^{\ast}\in \mathbb{R}^{d^\ast}$ such that
\begin{equation*}
    \ell^{\ast}\in  \displaystyle{\argmin_{\ell}} \ L_{\text{price},\alpha}(\ell),
\end{equation*}
where the loss function $L_{\text{price},\alpha}(\ell)$ is given by
\begin{equation}\label{calibration:price1}
\begin{aligned}
	L_{\text{price},\alpha}(\ell)&:= \sum_{i=1}^{N}\Big(\Ss_{n}(\ell)_{t_i}-S_{t_{i}}\Big)^{2}+\alpha(\ell)\\
	&=\sum_{i=1}^{N}\Big(S_0+\ell_{\emptyset}\langle{e}_{1},\widehat{\X}_{t_{i}}\rangle+\sum_{0 < |I|\leq n-1}\ell_{I}\langle \tilde{e}_{I},\widehat{\X}_{t_{i}}\rangle-S_{t_{i}}\Big)^{2}+\alpha(\ell),
\end{aligned}
\end{equation}
where $\alpha$ denotes a fixed penalization function. An example for $\alpha$ is given by a convex combination of $L^{1}$ and $L^{2}$ penalizations, i.e., $\alpha(\ell)=\lambda\lVert \ell \lVert_{1}+(1-\lambda)\lVert\ell \lVert_{2}$, where  $\lambda \in [0,1]$. Note here that the number of parameters  corresponds to $d^{\ast}$, i.e.~the dimension of the signature truncated at level $n-1$. This is due to the representation of Corollary~\ref{cor1}, as also seen from the second equality in \eqref{calibration:price1}.

\begin{example}\label{ex3} Depending on the choice of the primary process, the time series for $X$ may not be directly readable from the market. This is for instance the case  if we assume $X$ to be a vector of correlated Brownian motions, whose trajectories have to be extracted from the observations of $S$. We discuss such a procedure in the following within the class of classical stochastic volatility models.

Suppose that the dynamics of the asset price process are described by a stochastic volatility model under the physical measure $\P$, e.g.,
\begin{equation*}
\begin{aligned}
	{\d}S_{t}&=\mu(S_{t},V_{t}){\d}t+g(S_t,V_{t}){\d}B_{t}^{\P}\\
	{\d}V_{t}&=h(S_t,V_{t}){\d}t+\sigma(S_t,V_{t}){\d}W_{t}^{\P},
\end{aligned}
\end{equation*}
with  ${\d}[B^{\P},W^{\P}]_{t}=\rho {\d}t$ where $\rho\in[-1,1]$ and $\mu, h, g, \sigma$ functions from $\mathbb{R}^2$ to $\mathbb{R}$ such that $g$ and $\sigma$ are strictly positive. Choosing  $d=2$, we aim to estimate the trajectories of a primary process $X_{t}:=(B_{t}^\Q,W_{t}^\Q)$ specified by
\begin{equation}\label{eq:Wq}
	B_{t}^{\Q}=\int_{0}^{t}\frac{\mu(S_{s},V_{s})}{g(S_s,V_{s})}{\d}s+B_{t}^{\P},\qquad \qquad W_{t}^{\Q}=\int_{0}^{t}\frac{h(S_s,V_{s})}{\sigma(S_s,V_{s})}{\d}s+W_{t}^{\P}.
\end{equation}
Note that under sufficient integrability conditions on the market prices of risk $\mathbb{Q}$ is an equivalent local martingale measure, under which $(B^{\Q}, W^{\Q})$ are correlated $\Q$-Brownian motions and under which the dynamics of $(S,V)$ are of the following form
\begin{equation*}
\begin{aligned}
	{\d}S_{t}&=g(S_t,V_{t}){\d}B_{t}^{\Q},\\
	{\d}V_{t}&=\sigma(S_t,V_{t}){\d}W_{t}^{\Q}.
\end{aligned}
\end{equation*}
Note that other choices of equivalent local martingale measure $\widetilde{\mathbb{Q}}$ are of course possible. This then amounts to specify the second  Brownian motion $W_{t}^{\widetilde{\Q}} $ (under the assumption that we stay in a Markovian framework) as
\[
W_{t}^{\widetilde{\Q}}=\int_{0}^{t}\frac{h(S_s,V_{s})- h^{\widetilde{\Q}}(S_s,V_s)}{\sigma(S_s,V_{s})}{\d}s+W_{t}^{\P}.
\]
Under this measure the dynamics of $V$, then read as
\begin{equation*}
	{\d}V_{t}=h^{\widetilde{\Q}}(S_t,V_t)\d t +\sigma(S_t,V_{t}){\d}W_{t}^{\widetilde{\Q}}.
\end{equation*}
A simple choice is of course to choose $h^{\widetilde{\Q}}=h$ so that $W^{\widetilde{\Q}}=W^{\P}$.
We consider in the applications below always the case  $h^{\widetilde{\Q}}=0$, i.e. $W^\Q$ as defined in \eqref{eq:Wq}.  

In order to extrapolate  the trajectories of  $(B^\Q, W^\Q)$, we estimate the spot quadratic variation of the price and of the volatility process (see e.g.~\cite{JP:11} for different types of pathwise spot covariance estimators in general semimartingale contexts), to get an estimator of $\Sigma_{11}:=g(S_t,V_t)^2$ and $\Sigma_{22}:=\sigma(S_t,V_t)^2$, which in turn allows to compute
	\begin{equation}\label{eq:bm_estimation}
		{\d}B_t^{\Q} = \frac{{\d}S_t}{\sqrt{\Sigma_{11,t}}}, \qquad {\d}W_t^{\Q}= \frac{{\d}V_t}{\sqrt{\Sigma_{22,t}}}.
	\end{equation} 
	Observe that assuming $g$ and $\sigma$ to be strictly positive we get that $\sqrt{\Sigma_{11,t}}=g(S_t,V_t)$ and $\sqrt{\Sigma_{22,t}}=h(S_t,V_t)$. This is important to guarantee that $\d S_t=g(S_t,V_t) \d B_t^\Q$ (as requested in Corollary~\ref{cor1}) and $\d V_t=h(S_t,V_t) \d W_t^\Q$.

 Having extracted the trajectories of $B^\Q$ and $W^\Q$ we can readily compute the trajectories of their time extended signature.
If one wants to work with a different equivalent local martingale measure $\widetilde{\Q}$ parameterized via $h^{\widetilde{\Q}}(S_s,V_s)$ as above, then $W_{t}^{\widetilde{\Q}}$ can obtained via
\[
\mathrm{d} W_{t}^{\widetilde{\Q}}= \frac{{\d}V_t}{\sqrt{\Sigma_{22,t}}}
- \frac{ h^{\widetilde{\Q}}(S_t,V_t)}{\sqrt{\Sigma_{22,t}}}{\d}t.
\]
\end{example}

\subsubsection{Regression using spot volatility data}
Let us now discuss the second method, which works
when the primary process $X$ is given by a vector of correlated Brownian motions similarly as in Example \ref{ex3}. Indeed, we then can also calibrate to the time-series data of the spot volatility of the trajectories of $S$. The same task was tackled in \cite{PSS:20} for the one-dimensional case with one driving Brownian motion as primary process and with a lead-lag extension instead of a time extension. In our framework the calibration to the spot volatility consists in finding $\ell^{\ast}\in \mathbb{R}^{d^\ast}$ such that
\begin{equation*}
    \ell^{\ast}\in  \displaystyle{\argmin_{\ell}} \ L_{\text{vol},\alpha}(\ell),
\end{equation*}
where the loss function $L_{\text{vol},\alpha}(\ell)$ is now given by 
\begin{equation}\label{calibration:vol}
	L_{\text{vol},\alpha}(\ell):= \sum_{i=1}^{N}\Big(\ell_{\emptyset}+\sum_{0 < |I|\leq n}\ell_{I}\langle e_{I},\widehat{\X}_{t_{i}}\rangle-\sqrt{\frac{{\d}}{{\d}t}[S]_{t_{i}}}\Big)^{2}+\alpha(\ell),
\end{equation}
where $\frac{{\d}}{{\d}t}[S]$ denotes the spot quadratic variation process and  $\alpha$ a penalty function as in \eqref{calibration:price1}. Note that we still work with the signature model as of Corollary~\ref{cor1} (for $D=1$) exploiting its representation in form of \eqref{eq:martingalecond}.

\begin{remark}\label{rem42}
\begin{enumerate}
    \item As the spot volatility is not directly available but has to be estimated, the penalty function plays an important role in order to avoid overfitting to noisy data. Note also that not all  coefficients of the signature might be statistically significant. This includes in particular the case where  $d$ is large and the components of $\Xx$ are strongly correlated, and the case where the number of parameters to be calibrated is much higher  than the number of points we aim to fit.
    \item Let us stress that in the context of Example~\ref{ex3} the calibration to the price path and to the volatility process $g(S,V)$ are equivalent in the following sense. Once we have found $\ell^{\ast}$ solving \eqref{calibration:vol}, we obtain 
\[
[S]_t\approx \int_0^t\bigg(\ell_\emptyset^*+\sum_{0 < |I|\leq n} \ell_I^{\ast} \langle e_I,\mathbb{\widehat{X}}_s \rangle \bigg)^2\d s. \]
Since by definition
$S_t=\int_0^t g(S_s,V_s) \d B_s^\Q,$
$g(S_t,V_t)> 0$, and $[B^\Q]_t=t$, this yields
$$S_t\approx S_0+\int_0^t \bigg(\ell_\emptyset^*+\sum_{0 < |I|\leq n} \ell_I^{\ast} \langle e_I,\mathbb{\widehat{X}}_s \rangle \bigg)\d B^\Q_s,$$
which by Lemma~\ref{lem2} can be rewritten as
\[S_t\approx S_0 + \ell_\emptyset^* B_t^\Q+\sum_{0 < |I|\leq n} \ell_I^{\ast} \langle \tilde{e}_I,\mathbb{\widehat{X}}_t \rangle. \]
As we have the same expression for the parameters $\ell^*$ obtained via \eqref{calibration:price1}, the claim follows.
\end{enumerate}
\end{remark}

\subsubsection{Numerical examples} \label{num_ex_ts}

We now study the above regression tasks by means of several examples based on classical models. We start with a SABR-type and a Heston model.

\begin{example}
Consider two classical stochastic volatility models, namely a SABR-type model
$$
\begin{aligned}
	{\d}S_{t}&=\mu S_{t}{\d}t+S_{t}V_{t}{\d}B_{t}^{\P}\\
	{\d}V_{t}&=\kappa(\theta-V_{t}){\d}t+\sigma V_{t}{\d}W_{t}^{\P},
\end{aligned}
$$
and a Heston model
\begin{equation}\label{eqn30}
\begin{aligned}
	{\d}S_{t}&=\mu S_{t}{\d}t+S_{t}\sqrt{V_{t}}{\d}B_{t}^{\P}\\
	{\d}V_{t}&=\kappa(\theta-V_{t}){\d}t+\sigma \sqrt{V_{t}}{\d}W_{t}^{\P},
\end{aligned}
\end{equation}
where in both cases ${\d}[B^{\P},W^{\P}]_{t}=\rho {\d}t$ where $\rho\in[-1,1]$.

We now aim to approximate these models with signature models whose primary processes are $\Q$-Brownian motions extracted from the time series data of $S$ and $V$ as in Example~\ref{ex3}. The calibration is performed applying both procedures described above, namely the calibration to the price's trajectory and the calibration to the spot volatility trajectory. The corresponding loss functions are given by \eqref{calibration:price1} and \eqref{calibration:vol}, respectively. In both cases the truncation parameter is chosen to be $n=2$ and a Lasso's penalty function $\alpha(\ell)=10^{-5}\|\ell\|_1$ has been employed (a Ridge regression has led to similar
results). For the experiment, we select the models parameters
\begin{equation*}\label{parameters:sim}
    \{S_{0},V_{0},\mu,\kappa,\theta,\sigma,\rho\}:=\{1,0.08,0.001,0.5,0.15,0.25,-0.5\}.
\end{equation*}
Observe that in the  two models
the  process $V$ has a different meaning. Indeed, in the SABR it corresponds to the spot volatility while in the Heston it is the spot variance. As we choose the same parameters, this means that the volatility $V$ in the SABR model is smaller that the volatility $\sqrt{V}$ in the Heston model.

For both fit we extract 3 ticks of Brownian motions per day from trajectories that are assumed to be observable with a frequency of 4 seconds for 8 hours during 1 calendar year.

The calibrations have thus been performed over 1 calendar year of observations with 3 ticks per day ($T=1$ and $N=1095$). The obtained parameters $\ell^*$ are then tested simulating 0.5 years of new trajectories of the stochastic volatility models, extracting the corresponding $\Q$-Brownian motions, and computing the trajectories of $S_3(\ell^*)$ as in Corollary~\ref{cor1}. Results are illustrated in Figure~\ref{fig:regression_stochvol}.

\begin{figure}[H]
	\centering
	\begin{subfigure}[b]{0.64\textwidth}            
		\centerline{\includegraphics[width=\textwidth]{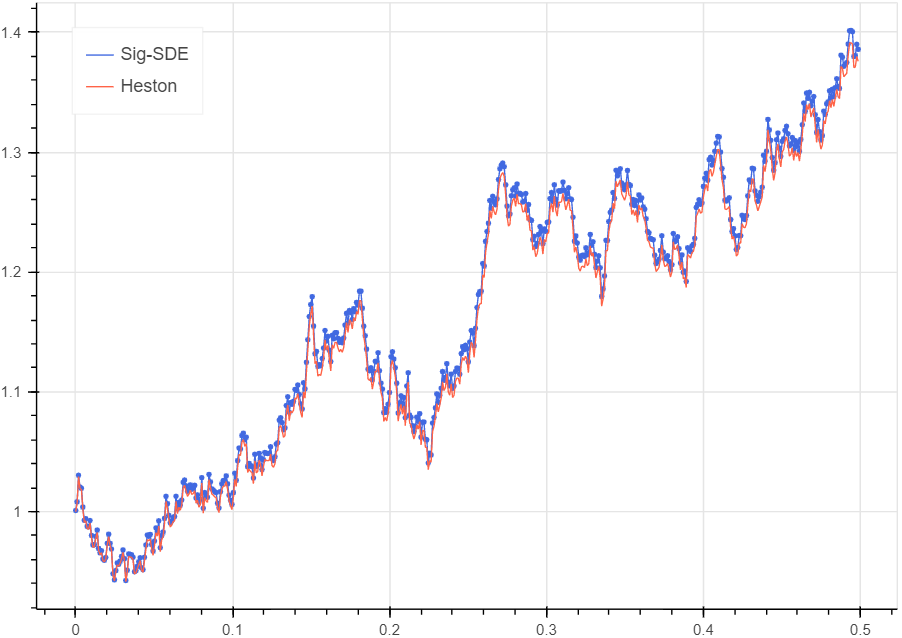}}
		\caption{Regression on the price of a Heston model as described in \eqref{calibration:price1}. }
		\label{fig:regression_price}
	\end{subfigure}

	\begin{subfigure}[b]{0.64\textwidth}
		\centering
		\centerline{\includegraphics[width=\textwidth]{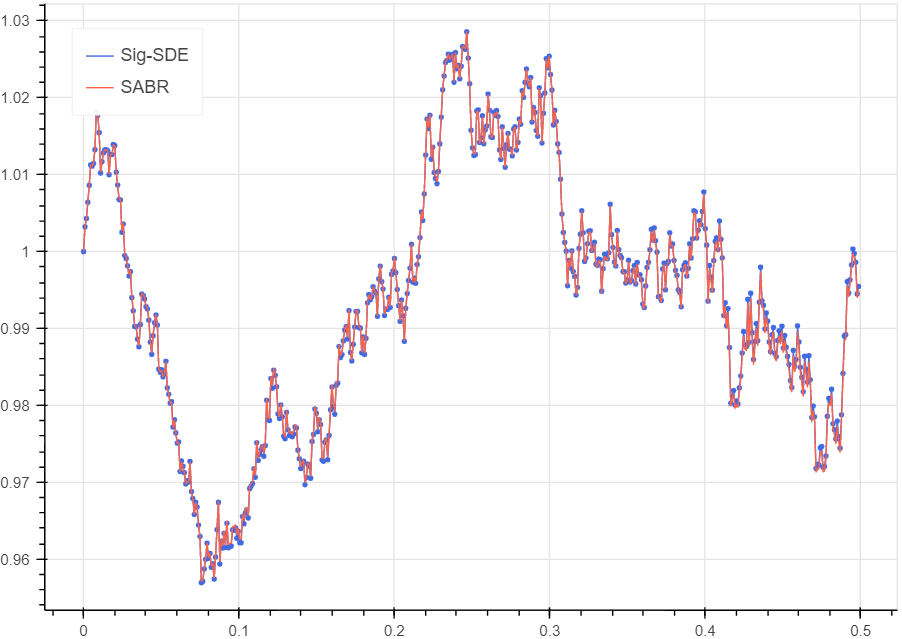}}
		\caption{Regression on the volatility of a SABR-type model as described in \eqref{calibration:vol}.}
		\label{fig:regression_vol}
	\end{subfigure}
	\caption{Out-of-sample of  half calendar year with training sample of one calendar year. }\label{fig:regression_stochvol}
\end{figure}
To assess the accuracy of the above calibration to time-series data we consider the mean squared error (MSE) between the predicted path and the observed path, i.e.
\[
\frac{1}{K}\sum_{i=1}^K (S_{n}(\ell^*)_{t_i}- S_{t_i})^2.
\]
In the following tables we state the in-sample MSE denoted  by $\text{MSE}_{\text{ex}}^{\text{in}}$ 
as well as the out-of-sample MSE denoted by 
$\text{MSE}_{\text{ex}}^{\text{out}}$. 
To justify the robustness of the estimated parameters the  out-of-sample MSEs are calculated as an average of the mean squared errors on 1000 different realizations of the simulated model.

\noindent Regression on the price as in \eqref{calibration:price1} under a Heston model:
\begin{center}
\renewcommand{\arraystretch}{2}
\begin{tabular}{||c | c||} 
 \hline
 $\text{MSE}_{\text{ex}}^{\text{in}}$ &
 $\text{MSE}_{\text{ex}}^{\text{out}}$\\ 
 \hline\hline
 $2.574\cdot 10^{-6}$ &
 $2.841\cdot 10^{-4}$ \\ 
 \hline
\end{tabular}
\end{center}
Regression on the volatility as in \eqref{calibration:vol} under the SABR-type model:
\begin{center}
\renewcommand{\arraystretch}{2}
\begin{tabular}{||c | c||} 
 \hline
 $\text{MSE}_{\text{ex}}^{\text{in}}$ &
 $\text{MSE}_{\text{ex}}^{\text{out}}$\\ 
 \hline\hline
 $3.204\cdot 10^{-6}$ &
 $1.060\cdot 10^{-6}$ \\ 
 \hline
\end{tabular}
\end{center}

It is worth mentioning that learning the volatility as in \eqref{calibration:vol} requires an estimation of the paths of  $(\sqrt{\frac{{\d}}{{\d}t}[S]_{t}})_{t\in[0,T]}$ from the price's trajectories. If such estimation is too noisy the procedure can yield unsatisfactory results.
One reason why this is here not visible in the out-of-sample MSEs  is that in the SABR model the volatility enters linearly in the price dynamics, and hence is easier to be learned via \eqref{calibration:vol}, which can be seen from the above tables. Note also that  the smaller volatility in the SABR models plays role in the order of magnitude of $\text{MSE}_{\text{ex}}^{\text{out}}$.
\end{example}

In the next example we consider a multivariate asset price process whose dynamics are given by a multi-dimensional Black-Scholes model.

\begin{example}[The multi-dimensional case]
Consider a $d$-dimensional Black-Scholes model under the physical probability measure $\P$, i.e., 
\begin{equation*}
	\d S_t=\mu^\top \operatorname{diag}(S_{t})\d t + \operatorname{diag}(S_{t})\sqrt{\Sigma}\d B_t^\P,
\end{equation*}
where $B^\P$ is a $d$-dimensional standard Brownian motion (with independent components), $\mu\in\mathbb{R}^{d}$,  $\Sigma \in \mathbb{S}_+^d$ and $\sqrt{\cdot}$ denotes the matrix square root. Similarly as in 
Example~\ref{ex3} we extract independent $\Q$-Brownian motions (i.e. $\mathbb{P}$ with the corresponding market price of risk)  which serve as primary process $X$ in the signature model, now  however of dimension $d > 1$.

In order to calibrate to the simulated price path $S$, we perform a Ridge regression using the prices' trajectories with penalty function $\alpha(\ell)=5\cdot10^{-5} \lVert \ell \lVert_{2}$. The truncation parameter is chosen to be $n=5$, corresponding to 46655 parameters to be calibrated. We considered the case of five positively correlated components, i.e.~$\Sigma_{ij}=\sigma_i\sigma_j\rho_{i,j}$ with $\sigma_i\in(0.15,0.34)$ and $\rho_{i,j}\in(0.4,0.9)$ for $i,j=1,\dots,5$,  with drifts $\mu_{i}\in(-0.003,0.001)$ for $i=1,\dots,5$. The calibration has been performed over 1 calendar year of observations with 3 ticks per day ($T=1$ and $N=1095$). The obtained parameters $\ell^*$ are then tested simulating 4 months of new trajectories of the Black-Scholes model, extracting the corresponding $\Q$-Brownian motions, and computing the trajectories of $S_6(\ell^*)$ as in Corollary~\ref{cor1}. Results are illustrated in Figure~\ref{fig:regression_multi}. The $\text{MSE}_{\text{ex}}^{\text{out}}$ on the single trajectories are of order $10^{-5}$.
\begin{center}
\begin{figure}
	    \centering	\includegraphics[width=0.64\textwidth,clip=false]{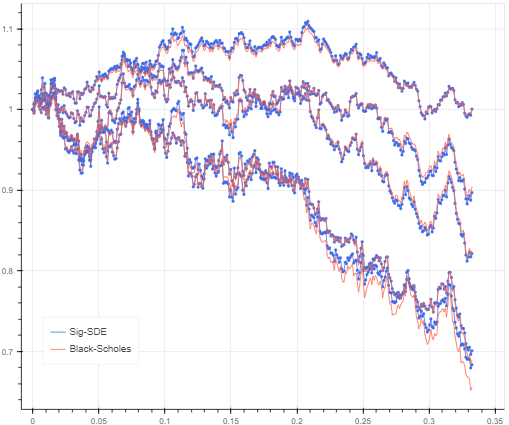}
	\caption{Regression on the price of a 5-dimensional Black-Scholes model as described in \eqref{calibration:price1}. Out-of-sample of 4 months with training sample of 1 calendar year.}
	\label{fig:regression_multi}
\end{figure}
\end{center}
\end{example}

\subsection{Calibration to option prices}\label{sec:calibration_options}

We now pass to the more common way of calibrating models in mathematical finance, namely to calibration to call (and put) option prices, with the goal to match their implied volatilities as well as possible.

Throughout we here assume that there exists a pricing measure $\Q$ and that the primary process $\Xx$ is a  vector of correlated $\Q$-Brownian motions.
The goal of this approach is to choose the parameters $\ell$ in order to fit the prices of options on $\Ss$ available on the market. In the context of Sig-SDEs, a similar approach has been considered  in \cite{PSS:20}. There the authors however only consider option prices generated from a Black-Scholes model  and  maturities longer than approximately 5 months. To our knowledge, in the context of signature based models for asset prices, we are the first ones dealing with market data, thus including shorter maturities and realistic volatility smiles. 
Suppose we are given prices of $N$ call options $C^{\ast}(T_{1},K_{1}),\dots,C^{\ast}(T_{N},K_{N})$  with maturities $T_{i}$ and strikes $K_{i}$. Then in spirit of \cite{CB:04} (see also \cite{CKT:20})  we translate the current calibration task  to finding $\ell$ such that
\begin{equation*}
    \ell\in \argmin_{\ell}\ L_{\text{option}}(\ell),
\end{equation*}
where 
\begin{equation}\label{calibration:options}
    L_{\text{options}}(\ell)=\sum_{i=1}^{N}\gamma_{i}\Big(C^{\ast}(T_{i},K_{i})-C^{\text{model}}(T_{i},K_{i},\ell)\Big)^{2},
\end{equation}
with $\gamma_{i}$ Vega weights and $C^{\text{model}}(T_{i},K_{i},\ell)$ denoting the price of the option under the signature model with parameters $\ell$, maturity $T_{i}$ and strike $K_{i}$.  

\begin{remark}
The loss function described in \eqref{calibration:options} can be adapted depending on the available data.
A common approach is for instance to weigh the options by their bid-ask spreads, see for instance \cite{CB:04}. This reflects the relative importance of reproducing different option prices precisely. In the present work we shall employ Vega weights since bid-ask spreads are not always provided in the data and since they are well-suited to match the implied volatility surface.

The same calibration procedure would apply to any other type of option, e.g. path-dependent exotic ones: as we shall see below, the pricing of the latter with a signature model does not require any additional effort. For liquidity reasons we however  decided to stick to vanilla options and calibrate to the corresponding implied volatility surface.
\end{remark}

To fix notation we briefly recall the definition of implied volatility, which is our goodness of fit criterium.

\begin{definition}
Let $C(T,K)$ be the price of a call option with maturity $T$ and strike price $K$ written on the asset $S$. The implied volatility of $C(T,K)$ is defined as the volatility $\sigma(T,K)$ that solves the equation
\begin{equation}\label{iv_equation}
    C^{BS}(T,K,\sigma(T,K))=C(T,K),
\end{equation}
where $C^{BS}$  Black-Scholes option price. Then $(\sigma(T,K))_{T,K}$ is called (implied) volatility surface, and  $(\sigma(T,K))_{K}$ is called (implied) volatility smile for each fixed maturity $T$.
\end{definition}

Once the optimal parameters $\ell^{\ast}$ are found via \eqref{calibration:options} we  then need to solve \eqref{iv_equation} numerically for $C(T,K)=C^\text{model}(T,K,\ell^\ast)$ to assess the  goodness of fit. Indeed we measure it in terms of the absolute error (in basis points) between the implied volatilities from the market and the model.

\subsubsection{Monte Carlo pricing}\label{mc-pricing}

We now address the problem of computing for a fixed strike $K\in\mathbb{R}$ and maturity $T>0$, the price $C^{\text{model}}(T,K,\ell)$ of the corresponding option under the signature model.

Before presenting our method based on Monte Carlo, recall that Section~\ref{sig-payoffs} provides a closed-form formula for the computation of sig-payoffs without the need to resort to Monte Carlo techniques. By the universal approximation theorem (Theorem~\ref{universality}), we also know that call options payoffs can be approximated arbitrarily well by sig-payoffs, or in this case even just by polynomials, on compacts. 
The combination of these two properties leads to the approach exploited by \cite{PSS:20} of splitting the computation of $C^{\text{model}}(T,K,\ell)$ in two parts: an approximation of the call (put) payoffs using sig-payoffs and pricing the latter with the closed formula proved in Section~\ref{sig-payoffs}. However, in order to achieve an acceptable error between the original call payoff and the sig-payoff one needs to choose a sig-payoff involving signature terms of quite high order. Indeed, such an approximation needs to be valid for every set of coefficients $\ell$ involved in the optimization procedure and thus on a very large compact set $K$.
From a computational point of view we observed that this sig-payoff or polynomial approximation is not tenable and we thus opted for abandoning this approach. 
Note however that sig-payoffs can still be used as control variates in Monte Carlo pricing techniques as we shall outline in Section \ref{sec:MCvar}.
For more details on the issues that can occur in view of sig-payoff or polynomial approximations  we refer to Section~\ref{sec:calib-sigpayoffs}.

For these reasons we therefore adopted a  Monte Carlo approach, which is due to the linearity of the model particularly tractable, making the whole calibration computationally feasible within a reasonable amount of time while providing highly accurate results.
For the Monte Carlo price we thus fix a number of samples $N_{MC}>0$  and approximate   $C^{\text{model}}(T,K,\ell)$ via
\begin{equation*}
    C^{\text{model}}(T,K,\ell)\approx \frac{1}{N_{MC}}\sum_{i=1}^{N_{MC}}(S_{n}(\ell)_{T}(\omega_i)-K)^{+}.
\end{equation*}
We stress again that this can be computed fast. Indeed,  by  the linearity of the model simulating $(S_{n}(\ell)_{T}(\omega_{i}))_{i=1}^{N_{MC}}$ boils down to the following steps:
\begin{itemize}
    \item simulate $(X_{t}(\omega_{i}))_{t\in[0,T]}$, which in the current setting are just trajectories for correlated Brownian motions, for each $i\in \{1,\ldots,N_{MC}\}$;
    \item compute $\langle e_{I},\widehat{\mathbb{X}}_{T}(\omega_{i})\rangle$ for all $i=1,\dots,N_{MC}$ and for all multi-indices $I$ such that $\lvert I\lvert \le n$;
    \item take linear combinations to compute $\langle \tilde{e}_{I},\widehat{\mathbb{X}}_{T}(\omega_{i})\rangle$ for all $i=1,\dots,N_{MC}$ and for all multi-indices $|I|\leq n-1$ as described in Lemma~\ref{lem2};
    \item retrieve $(S_{n}(\ell)_{T}(\omega_i))_{i=1}^{N_{MC}}$ via \eqref{eq:martingality}.
\end{itemize} 
Observe that 
the parameters $\ell$ only enter in the very last step which allows to precompute and store all other quantities. This is in contrast to other models where the parameters that need to be calibrated enter at each time point in the simulation steps of e.g.~an Euler scheme or more complicated schemes used for instance for rough volatility models. 

\begin{remark}
Since the calibration problem is usually ill-posed as inverse problem, there might exist more than one local minima of the loss function which is able to reproduce the prices on the market, and their implied volatility. We refer to \cite{CT:04} where the sensitivity, e.g.~in terms of initial starting points in gradient descent optimization algorithms, of the nonlinear least squares calibration problem in the case of exponential L\'evy models is discussed. In our setting it is worth mentioning  that since $S_{n}(\ell)_{T}$ is linear and the call payoffs are convex  we have a convex optimization problem whenever for all maturities and strikes
\begin{equation*}
    C^{\ast}(T,K)\le\frac{1}{N_{MC}}\sum_{i=1}^{N_{MC}}(S_{n}(\ell)_{T}(\omega_i)-K)^{+}.
\end{equation*}
Therefore the initial random parameter $\ell$ for the optimization can for instance be initialized according to this condition.
\end{remark}

\subsubsection{Variance reduction with sig-payoffs}\label{sec:MCvar}
Even though Monte Carlo pricing is fast since all essential quantities can be precomputed as explained above, we here discuss variance reduction techniques (see e.g.~\cite{G:04}) that can speed up the procedure even further. 
The idea is to introduce a control variate, i.e., a random variable $\Phi^{cv}$ such that:
\begin{equation*}
    \mathbb{E}_{\Q}[\Phi^{cv}]=0, \qquad \qquad \Var \big((S_{n}(\ell)_{T}-K)^{+}-\Phi^{cv}\big)< \Var \big((S_{n}(\ell)_{T}-K)^{+}\big).
\end{equation*}
An example of control variates used for pricing and calibrating neural SDE models can be found in  \cite{CKT:20,GSSSZ:20}, where $\Phi^{cv}$ is constructed from hedging strategies. A possible other choice of control variates for signature models are sig-payoffs. Indeed, one can use the pricing formula derived in Section \ref{sig-payoffs} to define:
\begin{equation*}
    \Phi^{cv}(T,K,\ell):=f_{\emptyset}+\sum_{0<\lvert J \lvert \le m} f_{J}\langle e_{J},\widehat{\mathbb{S}}_{n}(\ell)_{T}\rangle-\widetilde{P}_f(\ell,\E_\Q[\widehat \X_T]),
\end{equation*}
for some fixed $m>0$ where
\begin{equation*}
    (S_{n}(\ell)_{T}-K)^{+}\approx f_{\emptyset}+\sum_{0<\lvert J \lvert \le m} f_{J}\langle e_{J},\widehat{\mathbb{S}}_{n}(\ell)_{T}\rangle,
\end{equation*}
for a wide range of $\ell$ and with high probability. This can be done by performing a linear regression to obtain the coefficients $f$.
Alternatively, a polynomial approximation of the payoff's function can also be employed.

The properties of  $\Phi^{cv}$ then guarantee the accuracy of the approximation
\begin{equation*}
    C^{\text{model}}(T,K,\ell)\approx\frac{1}{N_{MC}}\Big(\sum_{i=1}^{N_{MC}}(S_{n}(\ell)_{T}(\omega_i)-K)^{+}-\Phi^{cv}(T,K,\ell)(\omega_i)\Big),
\end{equation*}
already for smaller values of $N_{MC}$.

\subsubsection{Model's performance} 
In the following we discuss the problem of minimizing the functional \eqref{calibration:options} using the Monte Carlo method as described in Section \ref{mc-pricing} to compute the model prices. We consider the model described in Corollary~\ref{cor1} for $\widehat \Xx_t=(t,B_t,W_t)$ as in Definition~\ref{def1}\ref{it1} for  two correlated Brownian motions $B$ and $W$ with correlation coefficient $\rho=-0,5$, and $D=1$.

As a first example we consider synthetic data, where the implied volatility surface we aim to fit is generated by a Heston model given by \eqref{eqn30}. We consider 7 maturities $(T_{k})_{k=1}^{7}$ ranging from 30 days to 2 years and 13 strikes $(K_{j})_{j=1}^{13}$ ranging from 80$\%$ to 120$\%$ of the spot price. The truncation parameter is fixed to  $n=3$ and the number of Monte Carlo samples to $N_{MC}=10^6$. 
The results for the following two sets of parameters under the risk neutral measure
\begin{center}
\begin{tabular}{||c c c c c||} 
 \hline
  $\kappa$ & $\theta$ & $\sigma$ &$\rho$&  $V_{0}$\\ 
 \hline\hline
  0.1 & 0.1 & 0.4 & -0,5&  0.08\\ 
 \hline\hline
  0.2 & 0.3 & 0.5 & -0,5& 0.08\\ 
 \hline
\end{tabular}
\end{center}
are displayed in the first and in the second row of Figure~\ref{fig:calibration_mc_he1}, respectively.

\begin{center}
\begin{figure}[H]
    \centering
    \captionsetup{justification=centering}
    \centerline{\includegraphics[width=0.9\textwidth]{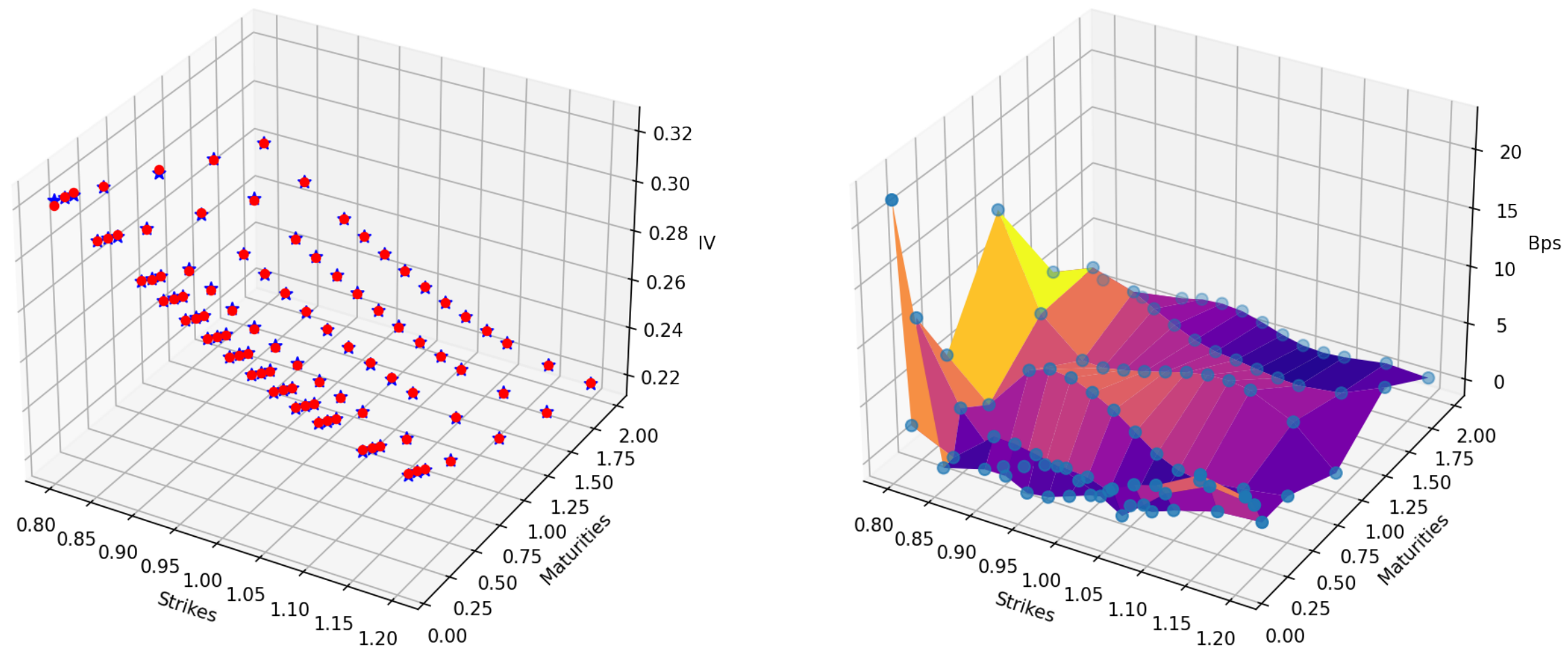}}

    \centerline{\includegraphics[width=0.9\textwidth]{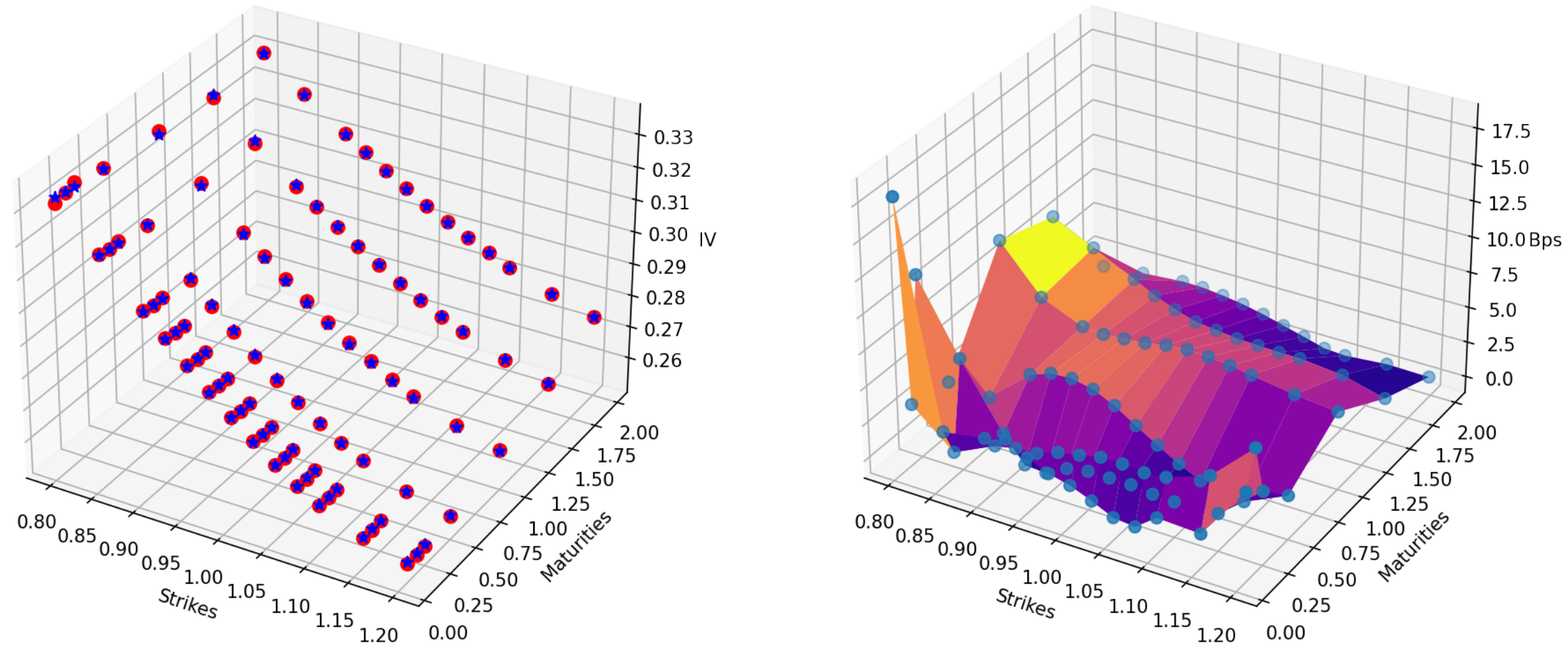}}
    \caption{On the left: blue stars correspond to the implied volatilities of the Heston models, red dots denote the calibrated implied volatilities of $S_{n}(\ell)$ with $n=3$ (13 estimated parameters). On the right: absolute errors between the two surfaces expressed in basis points (Bps).}
    \label{fig:calibration_mc_he1}
\end{figure}
\end{center}

We stress that these  calibrations to Heston generated implied volatility surfaces can take between 4 to 15 minutes  on a standard machine.

Let us now turn to  market data.
We consider the trading day 17/03/2021 for call options written on the S$\&$P 500 index. Our dataset provided by Bloomberg consists of 7 maturities $(T_{k})_{k=1}^{7}$, ranging from 30 days to 2 years, and 9 strike prices $(K_{j})_{j=1}^{9}$ for each maturity  which vary between 80$\%$ and 120$\%$ of the spot price. Again, the truncation parameter is fixed to  $n=3$ and the Monte Carlo's parameter to $N_{MC}=10^6$. The results are displayed in Figure~\ref{fig:calibraition_mc_constant2}.

\begin{center}
\begin{figure}[H]
    \centering
    \captionsetup{justification=centering}
    \centerline{\includegraphics[width=0.9\textwidth]{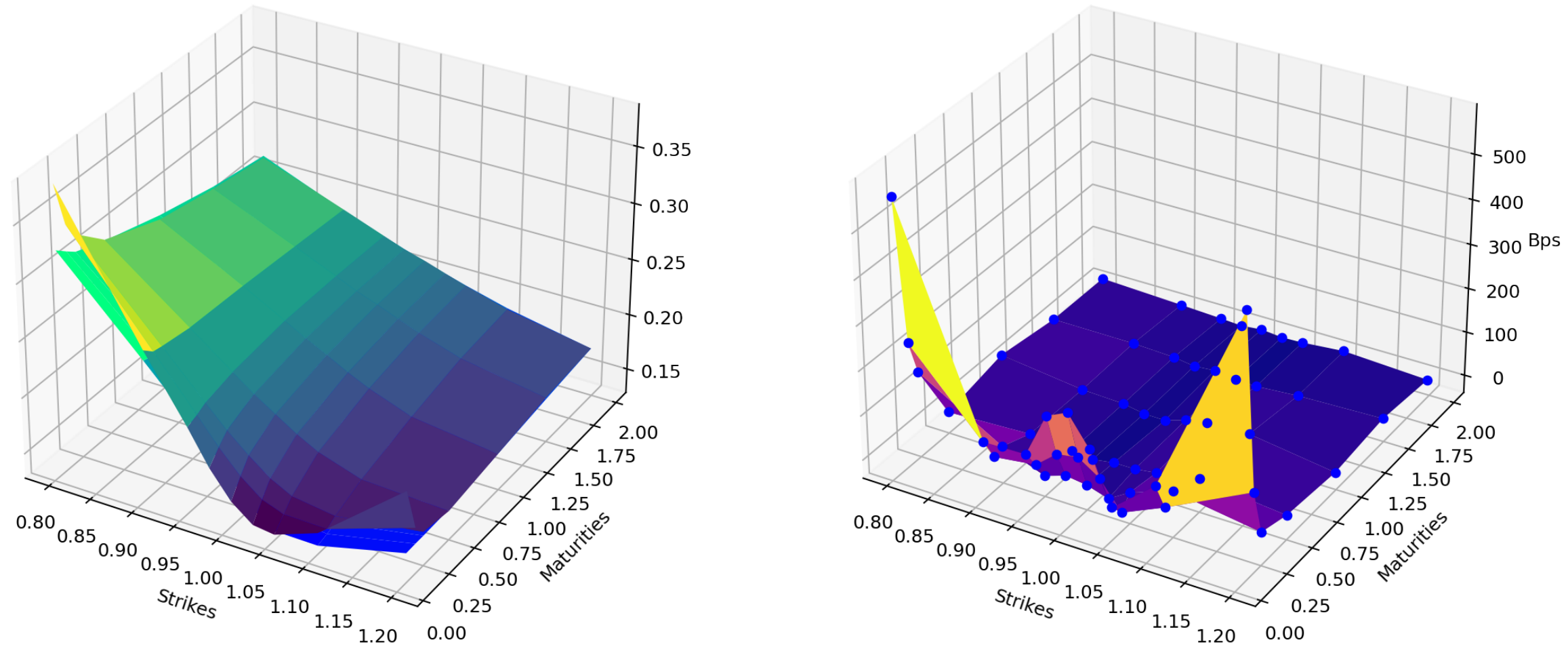}}
    \caption{On the left: the upper surface represents the implied volatility of the S$\&$P 500 index as of 17-03-2021,  the lower one is the calibrated implied volatility of $S_{n}(\ell^*)$ with $n=3$ (13 parameters). On the right: absolute error between the two surfaces (Bps).}
    \label{fig:calibraition_mc_constant2}
\end{figure}
\end{center}

From Figure \ref{fig:calibraition_mc_constant2} we notice that volatility smiles for short maturities have not been captured. In particular, as for many continuous models, the most problematic part consists in fitting the shortest maturities (see Chapter 3 and Chapter 7 of \cite{G:11}). 

To investigate the ability of the model to reproduce short maturity smiles we calibrate it using different loss functions penalizing outliers more severely. Precisely, we first calibrate the model using the procedure described above and we denote by $w_i$ the absolute error between the target implied volatility and the approximated one for maturity $T_i$ and strike $K_i$. Then, in spirit of generative-adversarial distances as for instance considered in \cite{CKT:20}, we define a new loss function
\begin{equation}\label{eqn13}
    L_{p,\alpha}(\ell)=\sum_{i=1}^{N}(\gamma_{i}+\alpha w_i)|C^{\ast}(T_{i},K_{i})-C^{\text{model}}(T_{i},K_{i},\ell)|^{p},
    \end{equation}
    depending on parameters $p$ and $\alpha$ that need to be chosen. By taking high values for $p$ and $\alpha$ we can approximate the
    sup-distance between the two price surfaces, i.e.,
    $$L_{\infty}(\ell)=\sup_{i=1,\ldots,N}|C^{\ast}(T_{i},K_{i})-C^{\text{model}}(T_{i},K_{i},\ell)|,$$
    without compromising differentiability with respect to $\ell$. 
    The result for different choices of the parameters $\alpha$ and $p$ but also the truncation level $n$ is displayed in Figure~\ref{fig:calibraition_mc_constantnew}. As can be guessed from the figure, although the maximal absolute error for maturities larger than 60 days is (almost) acceptable (96 Bps for $n=3$, $p=1000$, and $\alpha=500$ and 
    34 Bps for $n=4$ and $p=1000$, and $\alpha=500$), the absolute error for the shortest maturity and the far in and out of the money strikes are still above $270$ Bps and $395$ Bps respectively. Observe that the performance for $n=4$ is the best for every maturity larger than 60 days as well as for the at-the-money region of the shortest maturity.
    \begin{center}
\begin{figure}[ht]
    \centering
    \captionsetup{justification=centering}
    \centerline{\includegraphics[width=1.05\textwidth]{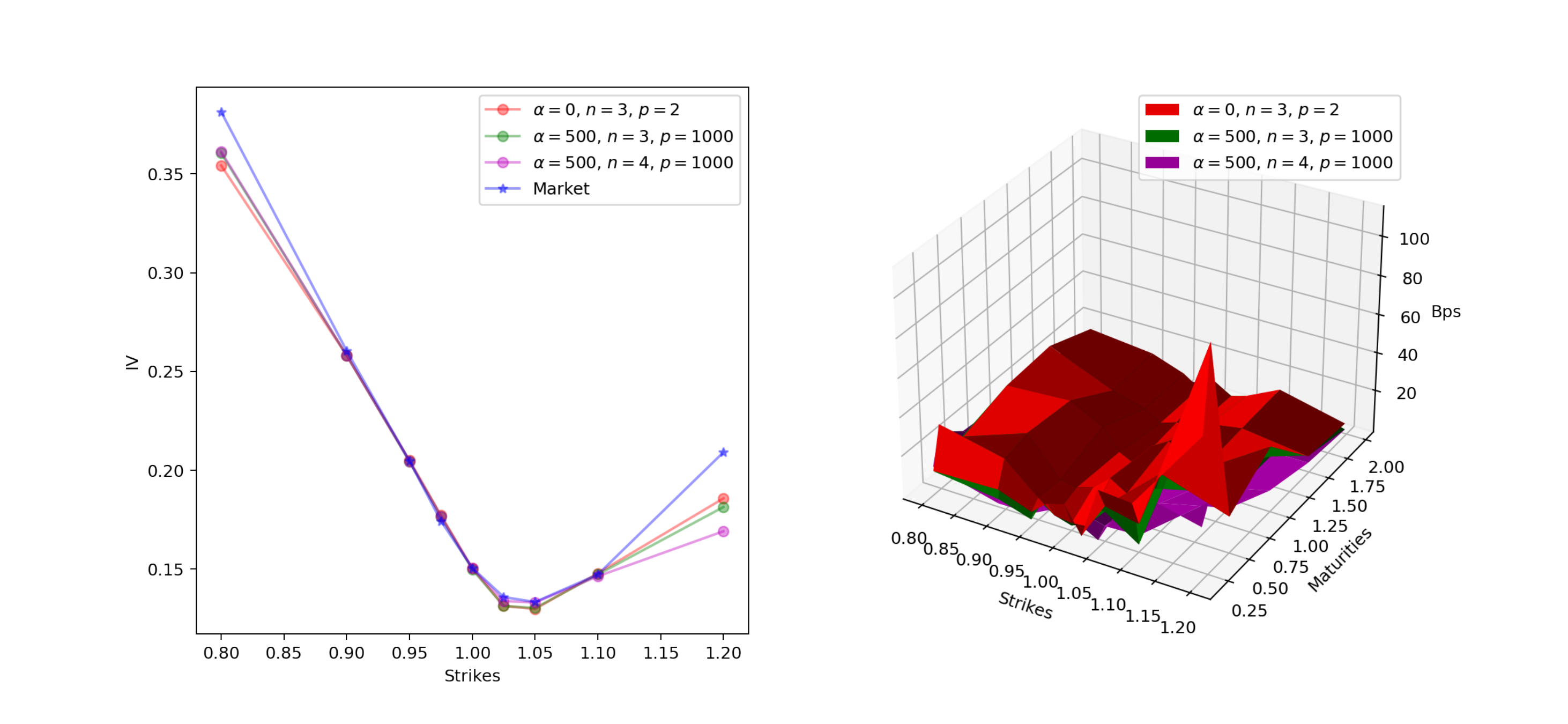}}
    \caption{Comparison between the calibrated implied volatility smiles for different parameters and the S$\&$P 500 index as of 17-03-2021 smile (in blue) at maturity $T_{1}=30$ days (on the left) and for maturities ranging from 60 days to 2 years (on the right). Calibration has been performed using \eqref{eqn13}.}
    \label{fig:calibraition_mc_constantnew}
\end{figure}
\end{center}
    Finally, as a further experiment, we calibrate the model to the shortest maturity alone and to every other maturity together. The first calibration is thus performed to $T_1=30$ days and the second one  to 6 maturities ranging from 60 days to 2 years. In both cases 9 strikes  ranging from 80$\%$ to 120$\%$ of the spot price are considered. The parameters are fixed to $n=2$, $p=2$, and $\alpha=0$ for the calibration to the first maturity and $n=4$, $p=300$, $\alpha=500$ for the remaining maturities. The result is displayed in Figure~\ref{fig1}. We can see that the absolute error is below 12 Bps for the first maturity and  45 for every other maturity, thus yielding a high accuracy.  The computational time needed for the calibration to the first smile is the range of a few minutes, the calibration to the remaining surface can take longer, which is due to the higher value of $n$ and the joint calibration to all maturities.
\begin{center}
\begin{figure}[H]
    \centering
    \captionsetup{justification=centering}
    \centerline{\includegraphics[width=0.9\textwidth]{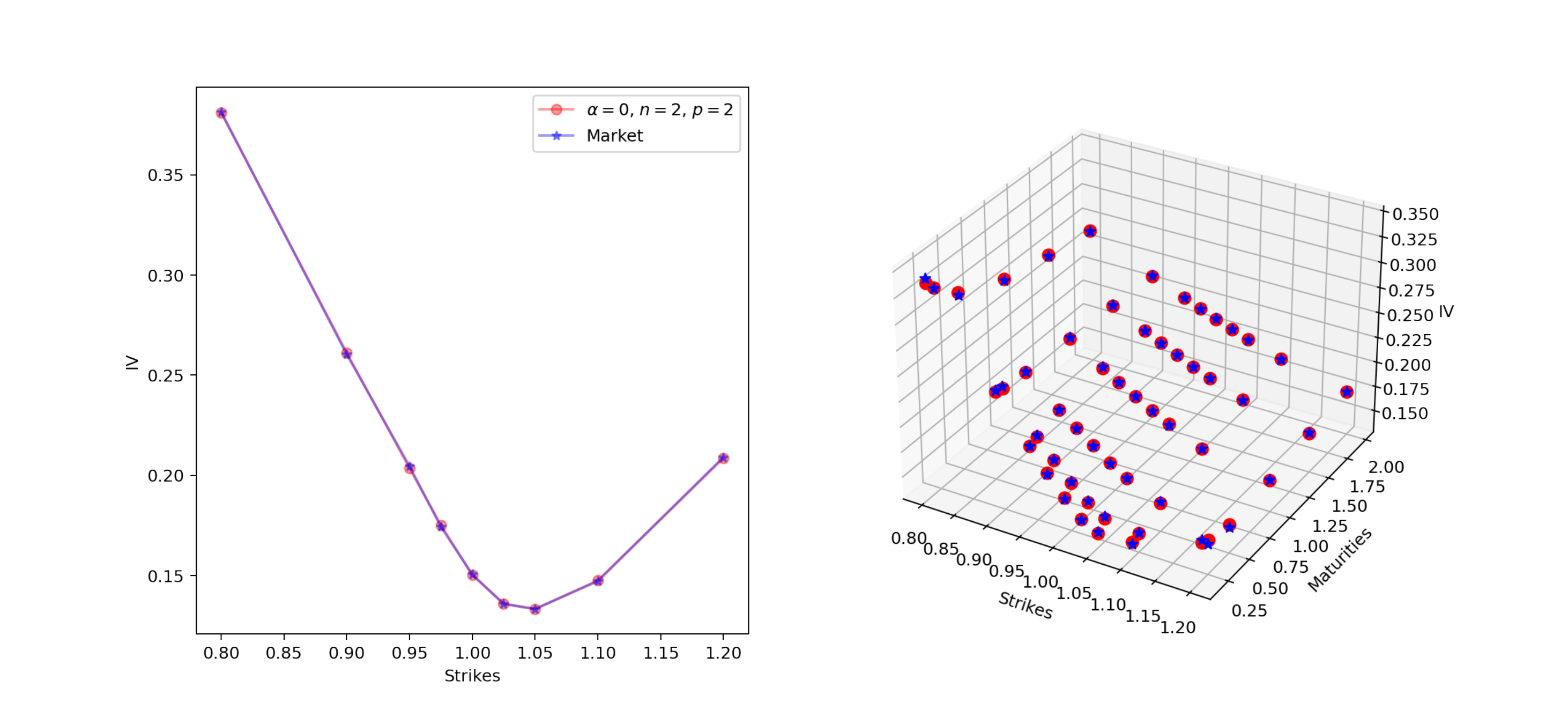}}
    \caption{Comparison between the calibrated implied volatility smiles and the SPX-500 Index as of 17-03-2021 smile (in blue) at maturity $T_{1}=30$ days (on the left) and for maturities ranging from 60 days to 2 years (on the right). Two different calibrations have been performed using \eqref{eqn13}.}
    \label{fig1}
\end{figure}
\end{center}  
These results suggest that introducing maturity-dependent parameters and  performing a slice-wise calibration to the individual smiles as for instance in \cite{CKT:20} or \cite{ GSSSZ:20} can be of interest to obtain both, an excellent accuracy and a low computational time. We investigate this further in Section \ref{sec:timevarying}.

\begin{remark}
Note that a signature model with a two dimensional primary process and $n=4$ as in the Figure~\ref{fig1} (right) has 121 parameters. This model is calibrated to 54 options prices. To investigate  possible over-fitting, we computed values of the implied volatility surface for out-of-sample strikes and maturities. More precisely, we added in total 76 new points. The result is displayed in Figure~\ref{plotdeisogni} where the out-of-sample points are displayed in red and the surface corresponds to the one of Figure~\ref{fig1} (right) for maturities from 60 days to 2 years. This shows that the model exhibits highly promising generalization features and is thus well suited for fast arbitrage-free smile inter- and extrapolation.
\begin{center}
\begin{figure}[H]
    \centering
    \captionsetup{justification=centering}
    \centerline{\includegraphics[width=0.6\textwidth]{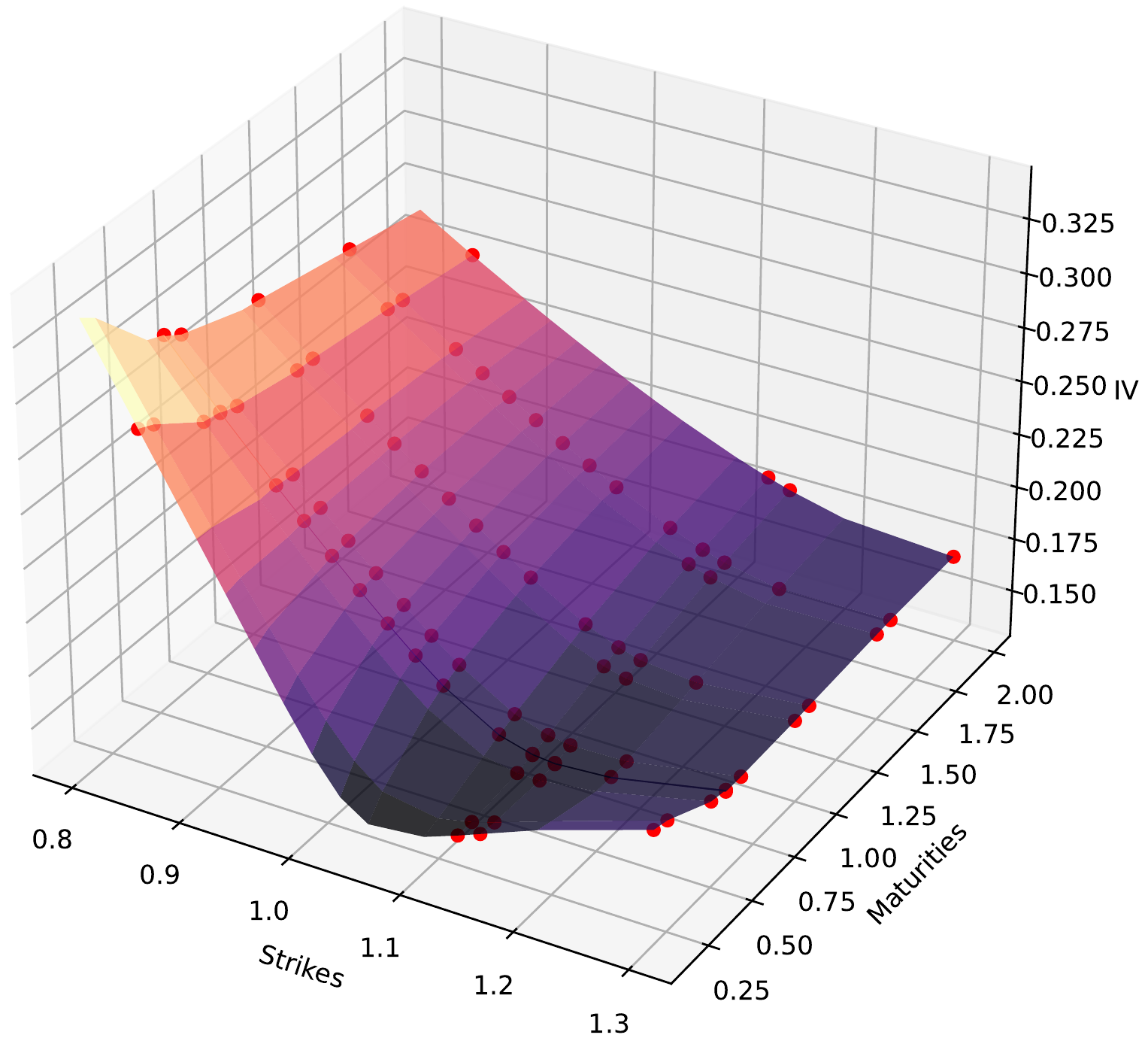}}
    \caption{Implied volatility surface generated by $S_4(\ell^*)$ with $\ell^*$ as in Figure~\ref{fig1} (right). Red dots denote out of sample points.}
    \label{plotdeisogni}
\end{figure}
\end{center}  
\end{remark}
\subsubsection{Time-dependent parameters}\label{sec:timevarying}
Motivated by the above results, 
we propose in the following a  variant of the initial model where the parameters $\ell$ are allowed to depend on time such that we can perform a slice-wise calibration.
 Set $\widehat X$ as in Definition~\ref{def1}\ref{it1}
and let $\{T_{1},\dots,T_{N}\}$ with $N>1$ be a set of fixed times, usually corresponding to the available maturities, and set
\begin{equation}\label{matk-model}
    \overline  S_{n}(\ell,\ell^{\text{corr}})_{t}:=
    S_{n}(\ell)_t+R_{n}(\ell^{\text{corr}})_t,
\end{equation}
where $S_{n}(\ell)$ is given as in Corollary~\ref{cor1} for $D=1$ and 
$$  R_{n}(\ell^{\text{corr}})_t
=\sum_{j=1}^{N-1}\sum_{| I | \le n-1}\ell_{I}^{\text{corr},j}1_{\{t\geq T_j\}}
    \langle \tilde{e}_{I},\widehat{\mathbb{X}}_{t}-\widehat{\mathbb{X}}_{T_{j}}\rangle.$$
    Observe that $R_{n}(\ell^{\text{corr}})_t=0$ for each $t\leq T_1$. This choice is motivated by the observation that $S_{n}(\ell)$ performs very well if calibrated on the first maturity only (see Figure~\ref{fig1}). On every other $T_k$ we have
    $$  R_{n}(\ell^{\text{corr}})_{T_k}
=\sum_{j=1}^{k-1}\sum_{| I | \le n-1}\ell_{I}^{\text{corr},j}
    \langle \tilde{e}_{I},\widehat{\mathbb{X}}_{T_k}-\widehat{\mathbb{X}}_{T_{j}}\rangle.$$
    Moreover, for any $t\geq0$ Lemma~\ref{lem2} yields
\begin{align*}
 \overline S_{n}(\ell,\ell^{\text{corr}})_{t}&=S_0+\int_0^t\sum_{| I | \le n-1}\bigg(\ell_{I}+\sum_{j=1}^{N-1}\ell_{I}^{\text{corr},j}1_{\{s\ge T_{j}\}}\bigg) \langle  e_I, \widehat{\mathbb{X}}_s \rangle \d X_s^1,
\end{align*}
showing that $ \overline S_{n}(\ell,\ell^{\text{corr}})$ is still continuous. We can also see that if $X^1$ is a local martingale the same is true for $\overline S_{n}(\ell,\ell^{\text{corr}})$.

\begin{remark}
The added correction term can be interpreted to take a similar role as the so-called leverage function  in local stochastic volatility models going back to \cite{L:02,RMQ:07}. 
Indeed, the leverage function is an adjustment to stochastic volatility models that
depends  on time (and state) and has the purpose to perfectly  fit the observed smiles, which by the underlying stochastic volatility model can usually not be achieved at the desired accuracy.
\end{remark}

Let us now describe the calibration procedure  for the adjusted model with time-dependent parameters.
 Suppose (for simplicity) that for each maturity $T_1,\ldots, T_N$ the price of call options with strikes $K_1,\ldots,K_M$ is known. The calibration is then performed recursively, adapting the loss function to the already calibrated parameters.
Precisely, the loss function used to calibrate $\ell$ is given by 
\begin{equation}\label{eqn14}
    L_{\text{options}}(\ell)=\sum_{j=1}^{M}\gamma_{1,j}\Big(C^{\ast}(T_{1},K_{j})-C^{\text{model}}(T_{1},K_{j},\ell)\Big)^{2},
\end{equation}
which is a version of \eqref{calibration:options}. Next, fix $k\in \{1,\ldots, N-1\}$ and let $\ell^{< k}:=(\ell, \ell^{\text{corr},1},\ldots,\ell^{\text{corr}, k-1})$ be the already calibrated coefficients.
The loss function used to calibrate $\ell^{\text{corr}, k}$ is given by 
\begin{align}\label{first-opt}
     L_{\text{options},k}(\ell^{\text{corr}, k})=\sum_{j=1}^{M}\gamma_{k,j}\Big(C^{\ast}(T_{k},K_{j})-C^{\text{model}}(T_{k},K_{j},\ell^{ <k},\ell^{\text{corr}, k})\Big)^{2}.
\end{align}
In both cases $\gamma_{k,j}$ denotes the Vega weight corresponding to maturity $T_k$ and strike $K_j$ and $C^{\text{model}}$ is computed using Monte Carlo  as outlined in Section~\ref{mc-pricing}, namely
$$C^{\text{model}}(T_{k},K_{j},\ell^{ <k},\ell^{\text{corr}, k})\approx\frac{1}{N_{MC}}\sum_{i=1}^{N_{MC}}(\overline S_{n}(\ell^{ <k},\ell^{\text{corr}, k})_{T_{k}}(\omega_i)-K_{j})^{+}.$$

To test this alternative approach we consider again the trading day 17/03/2021 for call options written on the S$\&$P 500 index, 7 maturities ranging from 30 days to 2 years, and 9 strike prices for each maturity  which vary between 80$\%$ and 120$\%$ of the spot price. The underlying process is chosen to be $X=(B,W)$ for two Brownian motions $B$ and $W$ with correlation $\rho=-0.5$. The truncation parameter is fixed to  $n=2$ and the number of Monte Carlo samples to $N_{MC}=10^6$.

In Figure~\ref{fig:comparison-ivs2} we report the results after minimizing \eqref{eqn14} and \eqref{first-opt} for all $k=1,\dots,6$. This shows that the signature model with time-dependent parameters perfectly fits each volatility smile. 
As illustrated in Figure~\ref{fig_skew} the model also manages to replicate
the observed term structure of the at-the-money (ATM) volatility skew, defined via
$ \psi(T):=\big\lvert\frac{\partial\sigma(T,K)}{\partial K} \big\lvert_{K=S_{0}},
$ for any maturity $T>0$,
very accurately. This is a feature
which can also be reproduced by rough volatility models, but classical stochastic volatility models rather generate a term structure of
the ATM skew that is constant for small maturities (see \cite{GJR:18}).

Finally, we stress the fact that  the calibration procedure is also very fast, as for each smile it only took about approximately 1 minute and 30 seconds. These results suggest that the signature model with time-dependent coefficients qualifies for very fast and highly accurate calibration fits.

In Figure \ref{fig:calibrated_parameters} we plot the values of the 91 calibrated parameters $\ell,\ell^{\text{corr}, 1},\ldots,\ell^{\text{corr}, 6}$. One can see that the corrections parameters get smaller and smaller. Note that we could have worked with a smaller number of corrections  leading to a similar performance in terms of accuracy. As shown at the beginning of the section one set of parameters is indeed enough to capture long maturities (e.g.~maturities longer than one year). 
\begin{center}
\begin{figure}[H]
    \centering
    \captionsetup{justification=centering}
    \centerline{\includegraphics[width=0.9\textwidth]{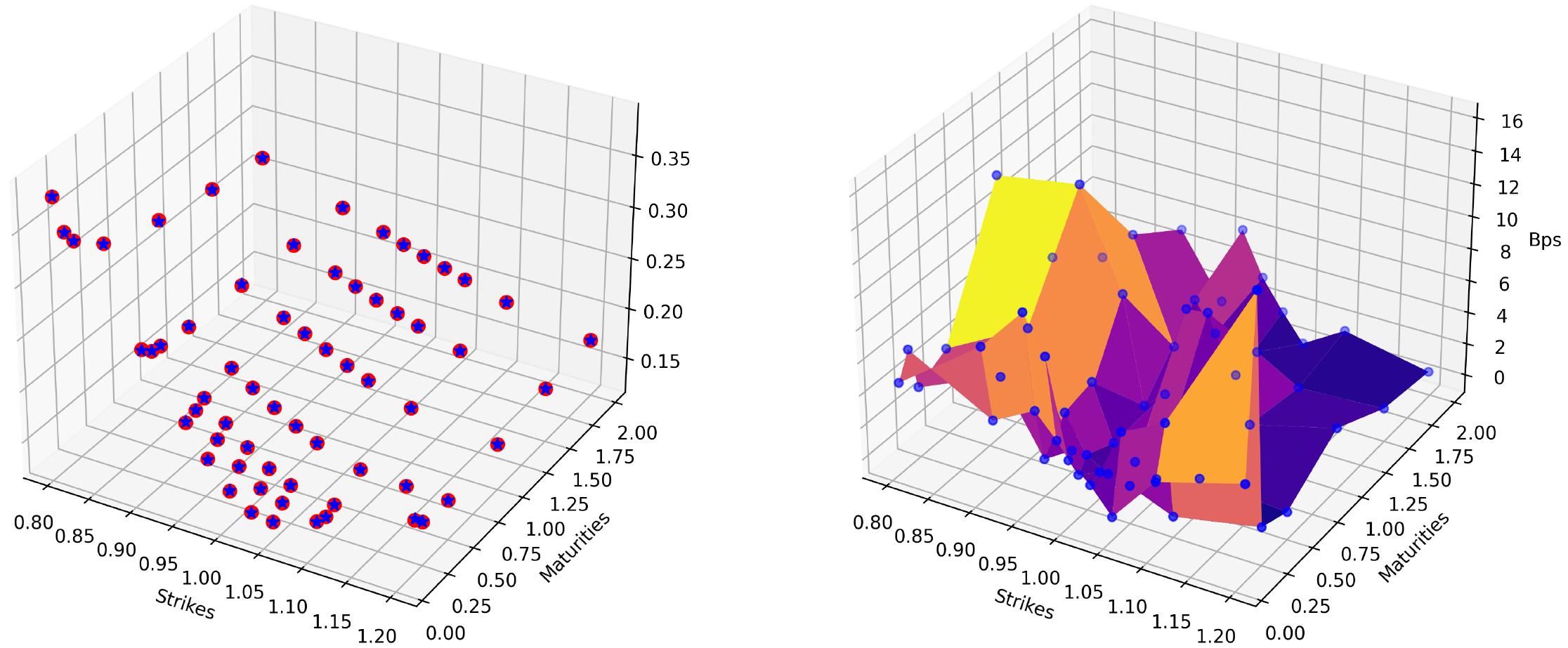}}
    \caption{
    On the left: blue stars denotes the implied volatilities on the market as of 17-03-21 on S$\&$P 500 index and the red dots denote the calibrated implied volatilities  with a signature model as in \eqref{matk-model} with $n=2$. On the right: absolute error between the two surfaces (Bps).}
    \label{fig:comparison-ivs2}
\end{figure}
\end{center}

\begin{center}
\begin{figure}[H]
    \centering
    \captionsetup{justification=centering}
    \centerline{\includegraphics[width=0.7\textwidth]{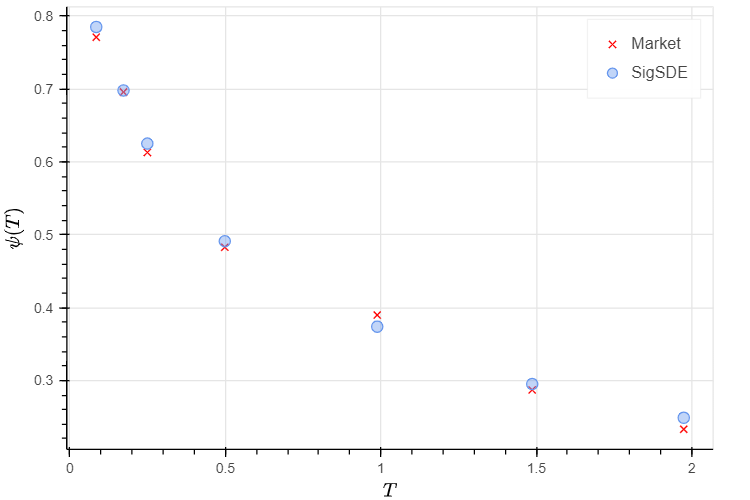}}
    \caption{Comparison between market term structure  of the at-the-money (red crosses) implied volatility skew $\psi$ and the calibrated one from the SigSDE (azure circles) minimizing \eqref{eqn14} and \eqref{first-opt}.}\label{fig_skew}
\end{figure}
\end{center}

\begin{center}
\begin{figure}[H]
    \centering
    \captionsetup{justification=centering}
    \centerline{\includegraphics[width=0.9\textwidth]{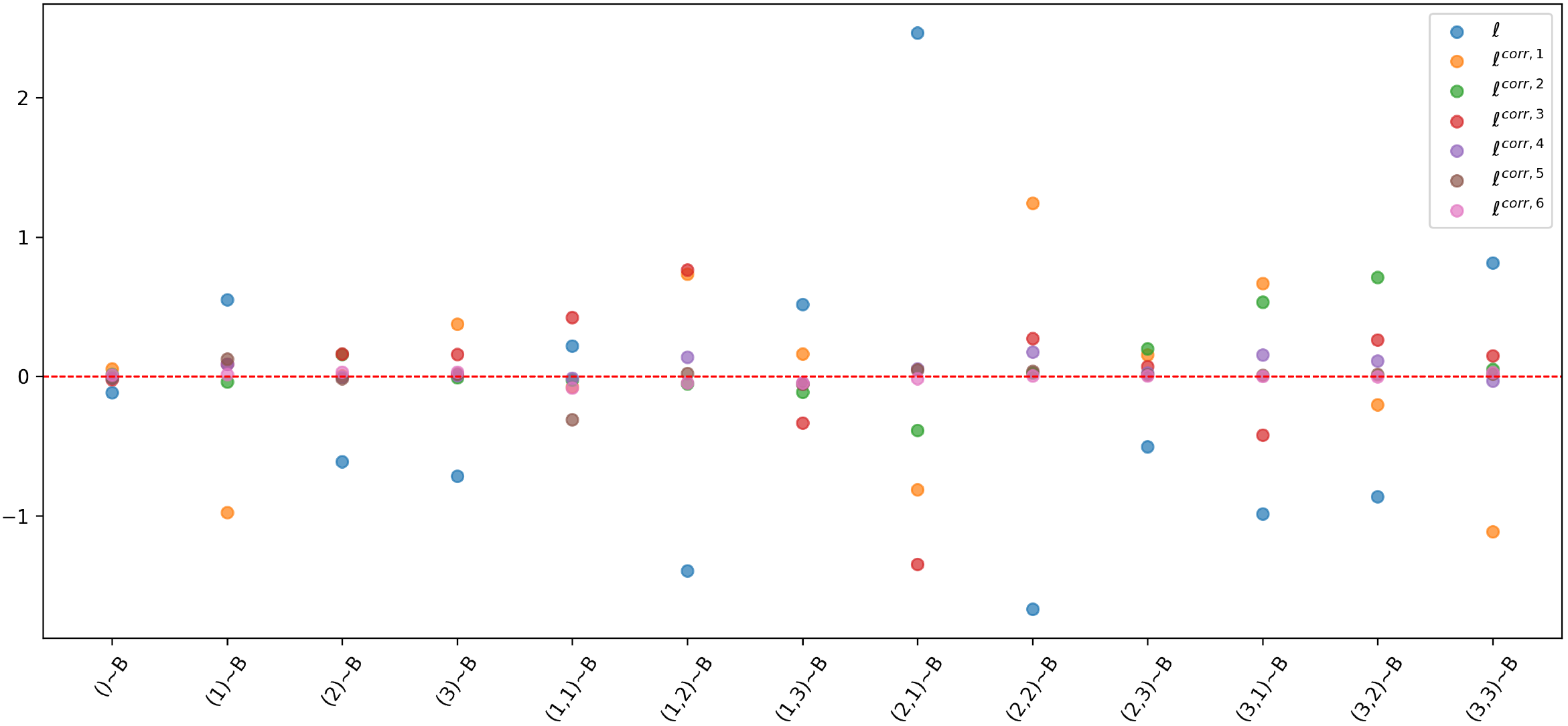}}
    \caption{Optimal parameters $\ell,\ell^{\text{corr}, 1},\ldots,\ell^{\text{corr}, 6}$ (consisting of 13 parameters each) found calibrating  to the volatility surface of  the market as of 17-03-21 on S$\&$P 500 index.}
    \label{fig:calibrated_parameters}
\end{figure}
\end{center}

\subsection{Joint calibration to time-series data and option prices}
In this section we combine the promising results obtained in Section \ref{sec:calibration_ts} and \ref{sec:calibration_options} to solve the following joint calibration problem. The main focus is to find a signature model that fits at the same time  both, a path realization coming from time series data as well as a volatility smile implied from option prices.
Let $\widehat X$ be as in Definition~\ref{def1}\ref{it1}, let $\{T_{1},\dots,T_{N}\}$ with $N>1$ be fixed times corresponding to the options' maturities and consider a signature model $S_{n}(\ell)$ as in Corollary~\ref{cor1}. Suppose that for each maturity $T_1,\ldots, T_N$ the prices of call options are available for several strikes,  and that we have time series data for the primary process and for the target model on a given time grid $\Pi$ at our disposal. Consider now the loss function used to calibrate $\ell$  given by 
\begin{equation*}
    L_{\text{joint}}(\ell)=\lambda L_{\text{options}}(\ell)+(1-\lambda)L_{\text{price},0}(\ell)
\end{equation*}
where $L_{\text{price},0}$ is given by  \eqref{calibration:price1} with $\alpha=0$ and times $t_i \in \Pi$,  $L_{\text{options}}$ is given by \eqref{eqn14}, and $\lambda\in[0,1]$.

\begin{remark}
When performing this joint calibration we implicitly assume to have a common primary process $X$ for both tasks.
We here suppose that $X$ is given by \eqref{eq:bm_estimation}, i.e.~we work with  Brownian motions under the local martingale measure $\Q$ specified in Example~\ref{ex3}. Note that under this measure $\Q$ the volatility\slash variance process $V$ is necessarily a local martingale.
Below we shall use a Heston model to 
generate both
the price trajectory  (under $\P$) and the option prices (under $\Q$). For the latter we specify the parameters to guarantee that $V$ is a (local) martingale under $\Q$.
One can of course also work with different local martingale measures as outlined  in Example \ref{ex3}.
\end{remark}
We consider an example where the traded asset $S$ follows Heston dynamics as described in \eqref{eqn30}, with the following parameters.
\begin{center}
\begin{tabular}{||c c c c c c c||} 
 \hline
 $\mu$ & $\kappa$ & $\theta$ & $\sigma$ &  $\rho$&$V_{0}$ & $S_{0}$\\ 
 \hline\hline
 0.001 & 0.8 & 0.1 & 0.55 & -0.5 & 0.12 & 1 \\ 
 \hline
\end{tabular}
\end{center}
under the physical measure $\mathbb{P}$. We fix $N=2$ and generate option prices at maturity $T_{1}=$3 months and $T_{2}=$1 year, under the set of parameters
\begin{center}
\begin{tabular}{||c c c c c c c||} 
 \hline
 $\mu$ & $\kappa$ & $\theta$ & $\sigma$ &  $\rho$&$V_{0}$ & $S_{0}$\\ 
 \hline\hline
 0 & 0 & 0 & 0.55 & -0.5 & 0.12 & 1 \\ 
 \hline
\end{tabular}
\end{center}
which ensures $S$ and $V$ to be local martingales under $\mathbb{Q}$. We consider additionally equi-spaced strike prices $\{K_{1},\dots,K_{9}\}$ with moneyness ranging between $70-130\%$ of the spot price. In Figure  \ref{fig:awesome_image3} we report the fit of the implied volatility smiles for $T_{1},T_{2}$  respectively, with less than 25 bps of absolute error. 
For the path calibration we extract daily Brownian motions from a Heston trajectory which is assumed to be observable with a frequency of 3 seconds for 8 hours during 9 months. The training sample consists of 3 months, namely 91 points and the out-of-sample performance in Figure \ref{fig:path_1mat}, is tested on 6 months. As hyper-parameters we choose $n=2$,  $\lambda=0.9$, and $N_{MC}=10^{6}$ trajectories for the Monte Carlo pricing. These plots  indicate that the model is also able to jointly fit time-series and option price data.

\begin{figure}[H]
    \centering
 \includegraphics[width=0.5\linewidth]{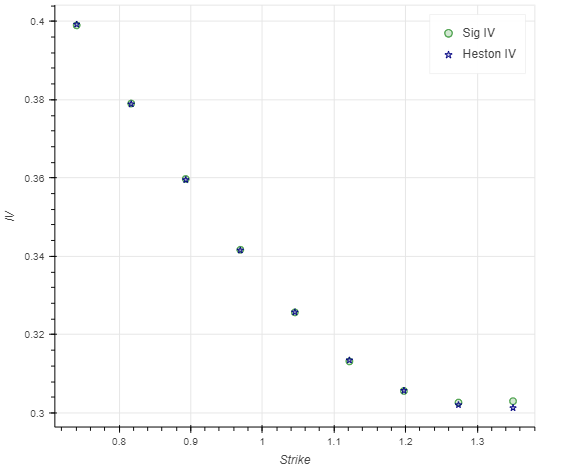}\includegraphics[width=0.5\linewidth]{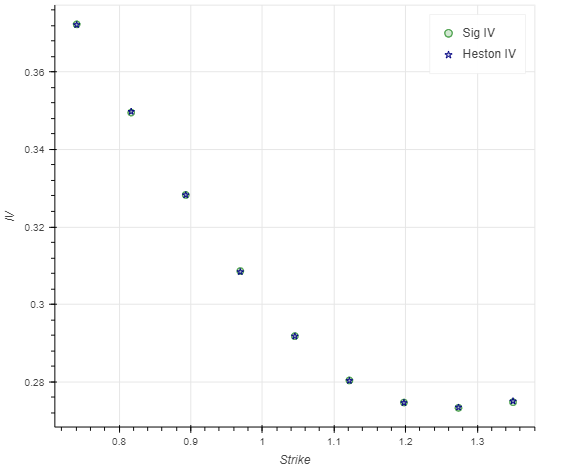}
  \caption{Comparison of implied volatilities for $T_1=$3 months and $T_{2}=$1 year.}
    \label{fig:awesome_image3}
\end{figure}
\begin{figure}[H]
    \centering
       \includegraphics[width=0.6\linewidth]{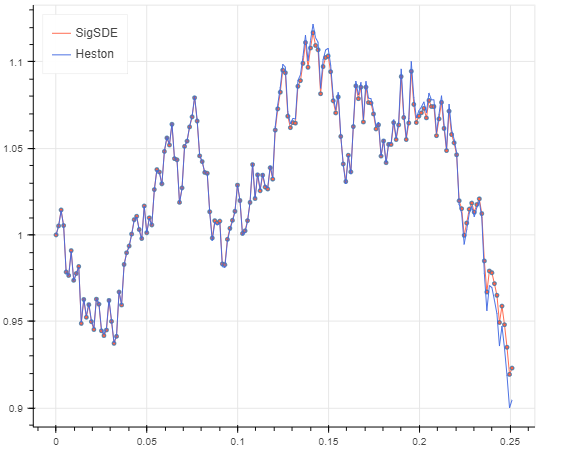}
    \caption{Out-of-sample path calibration performance, with MSE of $1.497\cdot 10^{-5}$.}
    \label{fig:path_1mat}
\end{figure}

\appendix

\section{Expected signature of correlated Brownian motions}\label{secBM}

A possible choice for the primary process $\Xx$ is given by a $d$-dimensional correlated Brownian motion $W$ with correlation matrix $\rho$. Since Assumption~\ref{ass1} is satisfied, we follow Definition~\ref{def1}\ref{it1} and set $\widehat W_t:=(t,W_t)$. The signature of $\widehat W$ up to time $t$ is denoted by $\widehat\W_t$.

 We explain now how to compute the expected signature of the extended Brownian motion $\widehat W$, i.e.~how to compute $\E[\langle e_I,\widehat\W_t\rangle]$ for each multi-index $I$. Several representations of related quantities including the signature of (independent) Brownian motions can be found in the literature, see for instance \cite{FW:03}, \cite{LV:04}, \cite{LN:15}, \cite{BDZN:21}. Observe that Theorem \ref{thm2} below can be seen as a particular case of the expected signature of a $d+1$-dimensional L\'evy process without  jumps (but with correlated Brownian motions), see \cite{FS:17} or of a multidimensional Gaussian process, see \cite{RFC:22}. Our approach is based on polynomial processes as introduced in \cite{CKT:12} and \cite{FL:16}. We   also refer to Section~4.5 in \cite{CS:21} for related computations for the signature of independent Brownian motions.
\begin{theorem}\label{thm2}
Consider a multi-index $I$ admitting the representation 
\begin{equation}\label{word}
	e_{I}=e_{0}^{\otimes k_{0}}\otimes e_{J_{1}}\otimes e_{0}^{\otimes k_{1}}\otimes e_{J_{2}}\otimes\cdots\otimes e_{0}^{\otimes k_{m}},
\end{equation}
for some $J_i\in\{1,\dots,d\}^{2h_{i}}$ and $h_i,k_i\in \N_0$. Then
\begin{equation*}
	\mathbb{E}[\langle e_{I},\widehat{\mathbb{W}}_{t}\rangle]=\frac{t^{\sum_{i=0}^{m}k_{i}+\sum_{i=1}^{m}h_{i}}}{\left(\sum_{i=0}^{m}k_{i}+\sum_{i=1}^{m}h_{i}\right)!}\bigg(\frac{1}{2}\bigg)^{\sum_{i=1}^{m}h_{i}}\prod_{i=1}^{m} \rho(J_{i}),
\end{equation*}
 where $\rho(J):=\prod_{k=1}^{|J|/2} \rho_{j_{2k-1},j_{2k}}$. Moreover, if $I$ does not admit representation \eqref{word} then $\E[\langle e_I,\widehat\W_t\rangle]=0$.
\end{theorem}
\begin{proof}
We first prove that the truncated signature $\widehat \W_t^N$ is a polynomial process in the sense of Definition~2.1 in \cite{FL:16}. By Lemma~2.2 in \cite{FL:16}, it suffices to show that for each $|I|\leq N$ it holds
\begin{equation}\label{eqn4}
\d\langle e_I,\widehat \W_t^N\rangle=b(\widehat \W_t^N)\d t+\sigma(\widehat \W_t^N)\d W_t
\end{equation}
for some linear maps $b,\sigma$. Observe that
\begin{align*}
	\langle e_{I}, \widehat \W^N_t \rangle&=\int_0^t\langle e_{I'}, \widehat \W^N_s \rangle\circ {\d}\langle e_{i_{|I|}}, \widehat W_s\rangle\\
	&=\int_0^t\langle e_{I'}, \widehat \W^N_s \rangle {\d}\langle e_{i_{|I|}}, \widehat W_{s}\rangle+\frac{\rho_{i_{{|I|}},i_{{|I|}-1}}}{2}1_{\{i_{{|I|}},i_{{|I|}-1}\neq 0\}}\int_0^t\langle e_{I''}, \widehat \W^N_s \rangle {\d}s\\
	&=\int_0^t\langle  1_{\{i_{{|I|}}=0\}}e_{I'} +\frac{\rho_{i_{{|I|}},i_{{|I|}-1}}}{2}1_{\{i_{{|I|}},i_{{|I|}-1}\neq 0\}} e_{I''}, \widehat \W^N_s \rangle \d s\\&\qquad+\int_0^t \langle e_{I'}, \widehat \W^N_s \rangle1_{\{i_{{|I|}}\neq 0\}} {\d}W_{s}^{i_{|I|}}.
\end{align*}
Since $|I'|,|I''|\leq N$ we can conclude that $\widehat\W^N$ satisfies \eqref{eqn4} and is thus a polynomial process.

Our next step consists in applying Theorem~3.1 in \cite{FL:16}. Observe that by  It\^o's formula we have that the generator $\Acal$ of $\widehat \W^N$ satisfies
$$\Acal\langle e_I,\fdot\rangle
:=\langle  1_{\{i_{|I|}=0\}}e_{I'} +\frac{\rho_{i_{|I|},i_{|I|-1}}}{2}1_{\{i_{|I|},i_{|I|-1}\neq 0\}} e_{I''},\fdot\rangle
=\sum_{e_{I_1}\otimes e_{I_2}=e_I}\langle e_{I_1},Q\rangle\langle e_{I_2},\fdot\rangle,
$$
for $Q:=e_0+\frac 1 2 \sum_{i,j=1}^d \rho_{i,j}e_i\otimes e_j$. Shortly, this can be written as
$\Acal\langle e_I,\cdot\rangle=\langle e_I,Q\otimes (\cdot)\rangle.$
Theorem~3.1 in \cite{FL:16} yields then that
\begin{align*}
    \E[\langle e_I,\widehat \W_t\rangle]
    &= \sum_{k=0}^\infty \frac 1 {k!}(t\Acal)^k\langle e_I,\fdot\rangle(\widehat \W_0)
    =\sum_{k=0}^\infty \frac {t^k} {k!}\langle e_I,Q^{\otimes k}\otimes \widehat \W_0\rangle
    =\sum_{k=0}^{|I|} \frac {t^k} {k!}\langle e_I,Q^{\otimes k}\rangle.
\end{align*}
Observe that $\langle e_I,Q^{\otimes k}\rangle=\sum_{e_{I_1}\otimes\cdots\otimes e_{I_k}=e_I}\prod_{i=1}^k\langle e_{I_i},Q\rangle$ is nonzero if and only if $I$ can be decomposed in multi-indices of length 1 whose elements are all 0, and multi-indices of length 2, whose elements are all strictly larger than 0, i.e.~if and only if $I$ satisfies \eqref{word}. If this is the case, by definition of $Q$ it holds
$$\E[\langle e_I,\widehat \W_t\rangle]
=\frac {t^{\sum_{i=0}^{m}k_{i}+\sum_{i=1}^{m}h_{i}}} {\left(\sum_{i=0}^{m}k_{i}+\sum_{i=1}^{m}h_{i}\right)!}\langle e_0,Q\rangle^{\sum_{i=0}^m k_i}\prod_{i=1}^m\langle e_{J_i},Q^{\otimes h_i}\rangle, $$
and since $\langle e_0,Q\rangle=1$ and $\langle e_{J_i},Q^{\otimes h_i}\rangle=(\frac 1 2 )^{h_1}\prod_{k=1}^{h_i}\rho_{j_{2k-1},j_{2k}}$ the claim follows.
\end{proof}
A combination of the result of Theorem~\ref{thm2} with Chen's identity (Lemma~\ref{lem5}) yields the following expression for the conditional expected signature.
\begin{lemma}
For each $s\leq t$ it holds
$$\mathbb{E}[\langle e_{I}, \widehat{\mathbb{W}}_{s+t}\rangle|\Fcal_s]
	=\sum_{e_{I_1}\otimes e_{I_2}=e_I}\langle e_{I_1}, \widehat{\mathbb{W}}_{s}\rangle
	\mathbb{E}[
	\langle e_{I_2}, \widehat{\mathbb{W}}_{t}\rangle].$$
\end{lemma}

\begin{proof}
Observe that applying Chen's identity and the definition of the tensor product we can compute
 \begin{align*}
	\mathbb{E}[\langle e_{I}, \widehat{\mathbb{W}}_{s+t}\rangle|\Fcal_s]
	&=
	\mathbb{E}[\langle e_{I}, \widehat{\mathbb{W}}_{s}\otimes\widehat{\mathbb{W}}_{s,s+t}\rangle|\Fcal_s]\\
		&=\sum_{e_{I_1}\otimes e_{I_2}=e_I}
	\mathbb{E}[\langle e_{I_1}, \widehat{\mathbb{W}}_{s}\rangle
	\langle e_{I_2}, \widehat{\mathbb{W}}_{s,s+t}\rangle|\Fcal_s]\\
			&=\sum_{e_{I_1}\otimes e_{I_2}=e_I}\langle e_{I_1}, \widehat{\mathbb{W}}_{s}\rangle
	\mathbb{E}[
	\langle e_{I_2}, \widehat{\mathbb{W}}_{s,s+t}\rangle|\Fcal_s].
	\end{align*}
Note that for each suitable integrand $H$ and each $u\in[0,t]$ it holds 
$$\int_s^{s+u}H_r\d \widehat W_r^i=\int_0^{u}H_{s+r}\d \widehat B_{r}^i,$$
where $\widehat B$ is the time extended Brownian motion given by $\widehat B_r:=\widehat W_{s+r}-\widehat W_s$. Since $B$ is independent of $\Fcal_s$ we can conclude that $\mathbb{E}[
	\langle e_{I_2}, \widehat{\mathbb{W}}_{s,s+t}\rangle|\Fcal_s]
	=\mathbb{E}[
	\langle e_{I_2}, \widehat{\mathbb{W}}_{t}\rangle]
	$, and the claim follows.
\end{proof}

\section{Calibration using sig-payoffs}\label{sec:calib-sigpayoffs}
We now discuss the procedure proposed in \cite{PSS:20} which  consists in calibrating the model to sig-payoffs that are intended to approximate the true payoffs. 
For general payoff functions $F$,  $$F:(\widehat\Ss_t)_{t\in[0,T]}\mapsto F\Big((\widehat\Ss_t)_{t\in[0,T]}\Big),$$
the first step consists in approximating $F$
with a sig-payoff such that
\begin{align}\label{eq:payoffapprox}
    F\Big((\widehat\Ss_t)_{t\in[0,T]}\Big)\approx f_\emptyset+\sum_{0<|J|\leq m}f_J\langle e_J,\widehat \S_{T}\rangle,
\end{align}
for some $m\geq 0$ and $f_\emptyset,f_J\in \R$. Here,  $\widehat \S_{T}$ denotes the signature of $\widehat \Ss=(t,S_t)$ at time $T$.
For our purposes, the approximation of the call payoffs  has to be accurate at least with high probability on the set of models
$$\{\widehat S_{n}(\ell)\colon \ell \text{ enters in the optimisation}\}.$$
Under the hypothesis that this can be achieved and proceeding as in Section~\ref{sig-payoffs},  the model call prices $C^{\text{model}}(T,K,\ell)$ (appearing in \eqref{calibration:options}) can then be (approximately) computed via
 $$C^{\text{model}}(T,K,\ell)=\E_\Q[(S_{n}(\ell)_T-K)^+]\approx f_\emptyset +\sum_{0<|J|\leq m}f_J \widetilde P_J(\ell,\E_\Q[\widehat\X_T]).$$
 The benefits from a computational perspective are immediate.

To achieve the approximation of the payoff functions, in \cite{LNP:20} and \cite{P:20} it is suggested to use an auxiliary stochastic process $(Y_t)_{t\geq0}$ playing the role of a market generator under $\Q$. A standard example for $Y$ is given by the Black-Scholes model 
$$\d
Y_t=\sigma Y_t \d W_t,\qquad Y_0=S_0,$$
where $W$ denotes a Brownian motion and $\sigma$ is sampled uniformly in the range of the implied volatility surface of the market, or as in \cite{P:20} between 5$\%$ and 40$\%$ of the spot price. We set again $\widehat Y_t:=(t,Y_t)$ and denote by $\widehat \Y$ the signature of $\widehat Y$.  The procedure to find the coefficients $f$ of \eqref{eq:payoffapprox} consists in a linear regression, fitting $N_{MC}>0$ realizations of the simulated sig-payoffs
 to the corresponding realizations of the simulated payoff. The corresponding loss function is given by
\begin{equation*}
L_{\text{payoff}}(f):= \sum_{j=1}^{N_{MC}}\bigg((Y_T(\omega_j)-K)^+-\Big(f_\emptyset+\sum_{0<|J|\leq m}f_{J}\langle e_{J}, \widehat\Y_{T}(\omega_{j})\rangle\Big)\bigg)^{2}.
\end{equation*}
In \cite{P:20} the author suggests to choose $N_{MC}=10^{4}$, $m=5$ and to sample $\sigma\sim \mathcal{U}([0.05,0.4])$.

We applied this procedure and then compared the call prices obtained under the Black Scholes model with $\sigma=0.25$ with the the approximation via sig-payoffs denoted by $F(T,K)$. For $T=0.5$, $Y_{0}=100$, and $K=120$, we obtain the following results:

\begin{center}
\renewcommand{\arraystretch}{2}
\begin{tabular}{||c |c|c| c||} 
 \hline
 $\mathbb{E}_{\mathbb{Q}}[(S_T-K)^+]$ & $\frac{1}{N_{MC}}\sum_{i=1}^{N_{MC}}(S_T(\omega_{i})-K)^+$ & $\mathbb{E}_{\mathbb{Q}}[F(T,K)]$ & $\frac{1}{N_{MC}}\sum_{i=1}^{N_{MC}}F(T,K)(\omega_{i})$\\ 
 \hline\hline
 1.5155 & 1.5145 & 1.6567 &  1.6550 \\ 
 \hline
\end{tabular}
\end{center}

Here, $\mathbb{E}_{\Q}[(S_T-K)^+]$ was computed via the Black-Scholes formula and $\mathbb{E}_{\Q}[F(T,K)]$ using the analytic expression of the expected signature of the time extended Black-Scholes model. We also compare the corresponding Monte Carlo prices using $N_{MC}=10^6$. We observed that the absolute percentage error between the Black-Scholes price and the sig-payoff option price is greater than $8\%$. The difference can become even more significant if the Black-Scholes model is replaced by  $S_{n}(\ell)$ for $\ell$ in a certain parameter range over which we optimize. So 
even if the model can provide an accurate fit to the sig-payoffs,
this does not mean that the true
implied volatility surface is well approximated, as the gap between call option prices and the approximate sig-option prices is simply too large. Therefore we opted for the procedure outlined in Section \ref{sec:calibration_options}.

\end{document}